\documentclass[12pt]{article}
\usepackage[utf8]{inputenc}
\usepackage{amsmath}
\usepackage{graphicx}
\usepackage{enumerate}
\usepackage{natbib}
\usepackage{longtable,threeparttablex} 
\usepackage{url} 
\usepackage{tikz}
\def\checkmark{\tikz\fill[scale=0.4](0,.35) -- (.25,0) -- (1,.7) -- (.25,.15) -- cycle;}

\usepackage{algpseudocode}
\usepackage{algorithm}
\usepackage{setspace,lipsum}
\usepackage{subfig}

\usepackage{booktabs}

\usepackage{sectsty}

\sectionfont{\large}
\subsectionfont{\large}

\usepackage{amsmath}
\usepackage{bbm}

\usepackage{xcolor}
\usepackage[english]{babel}

\usepackage{amsthm}
\usepackage{amssymb}

\theoremstyle{remark}
\newtheorem{remark}{Remark}[section]

\theoremstyle{plain}
\newtheorem{corollary}{Corollary}

\newtheorem{proposition}{Proposition}

\newtheorem*{otherthm}{Theorem}

\theoremstyle{definition}
\newtheorem{assumption}{Assumption}[section]
\theoremstyle{definition}
\newtheorem{example}{Example}[section]
\theoremstyle{definition}
\newtheorem{definition}{Definition}

\usepackage{graphicx}
\newcommand{\indep}{\raisebox{0.05em}{\rotatebox[origin=c]{90}{$\models$}}}

\renewenvironment{abstract}
 {\small
  \begin{center}
  \bfseries \abstractname\vspace{-.5em}\vspace{0pt}
  \end{center}
  \list{}{%
    \setlength{\leftmargin}{5mm}
    \setlength{\rightmargin}{\leftmargin}%
  }%
  \item\relax}
 {\endlist}

\def\spacingset#1{\renewcommand{\baselinestretch}%
{#1}\small\normalsize} \spacingset{1}

 \title{\vspace{-5ex} \bf \Large Estimating Structural Target 
 Functions using Machine Learning and Influence Functions}
\author{\normalsize{Alicia Curth}\thanks{Working paper.
	The results presented in this paper are part of research conducted by the first author (AC) for a dissertation submitted in partial fulfilment of the degree of Master of Science in Statistical Science at the Department of Statistics of the University of Oxford in September 2020. This work was conducted during a research internship of AC in the vanderschaar-lab at the University of Cambridge led by MvdS, and AC was supervised by MvdS and advised by AMA, postdoc with MvdS. AC is now a PhD student with MvdS. Correspondence to AC: amc253@cam.ac.uk }\\
    \normalsize University of Oxford\\ \vspace{-3ex}
    \normalsize University of Cambridge
    \and
    \normalsize Ahmed M. Alaa \\ 
    \normalsize UCLA 
    \and
    \normalsize Mihaela van der Schaar \\
    \normalsize University of
Cambridge\\
    \normalsize UCLA\\
    \normalsize The Alan Turing Institute
}
\date{\vspace{-5ex}}
\begin{document}
\maketitle

\begin{abstract}
We aim to construct a class of learning algorithms that are of practical value to applied researchers in fields such as biostatistics, epidemiology and econometrics, where the need to learn from incompletely observed information is ubiquitous. We propose a new framework for statistical machine learning of target functions arising as identifiable functionals from statistical models, which we call `IF-learning' due to its reliance on influence functions (IFs). This framework is problem- and model-agnostic and can be used to estimate a broad variety of target parameters of interest in applied statistics: we can consider any target function for which an IF of a population-averaged version exists in analytic form. Throughout, we put particular focus on so-called coarsening at random/doubly robust problems with partially unobserved information. This includes problems such as treatment effect estimation and inference in the presence of missing outcome data.
Within this framework, we propose two general learning algorithms that build on the idea of nonparametric plug-in bias removal via IFs: the `IF-learner' which uses pseudo-outcomes motivated by uncentered IFs for regression in large samples and outputs entire target functions without confidence bands, and the `Group-IF-learner', which outputs only approximations to a function but can give confidence estimates if sufficient information on coarsening mechanisms is available.  We apply both in a simulation study on inferring treatment effects.
\end{abstract}
\noindent%
\newpage
\spacingset{1.15} 
\section{Introduction}
Machine learning is increasingly transitioning from being a tool for prediction to taking over problems of interest classically within the domain of statistics, where the core focus is on \textit{estimation} and \textit{inference} instead of prediction. Because of the  inherent flexibility of machine learning models and their data-adaptive nature, this has opened up completely new possibilities for estimating possibly very complex structural target functions, instead of focusing only on low-dimensional parameters such as population averages. Arguably the most advanced example of this is the area of causal inference. In the last 10 years, there have been substantial advances in the development of methods for data-adaptive heterogeneous treatment effect estimation from both experimental and observational data using machine learning (see e.g. \citet{bica2020real} for a comprehensive overview). However, with some notable exceptions (e.g. \citet{alaa2018limits, Athey2016, Athey2019,Chernozhukov2018, kennedy2020optimal, wager2018estimation}) theoretical results making statistical guarantees for estimation and enabling inference are mostly lacking -- as is the case for most areas of machine learning \citep{Chernozhukov2018}.

Data-adaptive, machine-learning-based estimators have potential applications in many fields relying on applied statistics to empirically determine the effects of interventions, policies and treatments, and could be used to shift the focus from average treatment effect estimation towards flexibly investigating heterogeneity of effects across populations. One such potential application, which motivates the authors of this paper, is the move towards more personalised medicine and healthcare. In particular, the discovery of heterogeneous treatment effects in clinical trials could improve both the ability to treat patients and increase mechanistic understanding of underlying diseases. Clinical trials generally take place in small sample regimes, so it is important that estimators be \textit{efficient} in their use of data. Additionally, regulatory agencies require the quantification of \textit{statistical significance} of findings and possible \textit{bias}. If these factors can be accounted for, machine learning has immense potential to change the nature of clinical trials \citep{zame2020machine}. There are other applications in healthcare, where given that we can make guarantees on \textit{consistency} or worst-case performance, it might be enough to give personalised treatment recommendations because models have been trained on large (observational) data sets. Both types of settings -- learning from small experimental samples or large observational samples -- require at least some ability to give statistical guarantees on the performance of a learning algorithm. 

While the problem of heterogeneous treatment effect estimation has received considerable attention in related literature over the last few years, we take the stance that it is no different than most other problems in applied statistics. As long as a parameter is \textit{identifiable} from observed data and well-defined in terms of the underlying (unknown) statistical model, statistical estimation of any parameter has the same structure and inherent problems. This indicates that we can separate the issues of identification and estimation completely \citep{van2011targeted}. While it is thus necessary to carefully consider conditions under which treatment effect estimates can be interpreted as causal (such as those developed within the Neyman-Rubin Potential Outcomes  framework (\cite{neyman1923applications}, \citet{rubin1978bayesian}), or the graphical approach coined by \cite{pearl2009causality}),  from a statistical viewpoint there is, strictly speaking, no need for a separate statistical literature on treatment effect \textit{estimation} when identifiability is assumed.

With this in mind, we maintain a very general problem set-up in this paper and take a semiparametric, essentially assumption-free approach to estimating functions of continuous inputs -- arising as identifiable functionals from statistical models -- using plug-in estimation. On the one hand, this class of functions includes more well-studied parameters such as the conditional mean and the conditional outcome probability in standard regression and binary classification problems, respectively. On the other hand, it also includes problems where data is \textit{coarsened at random}, which have interesting \textit{doubly robust} structure that we can exploit for efficient and robust estimation. Such coarsening occurs in many scenarios of practical interest to empirical researchers in biostatistics, epidemiology and econometrics \citep{rubin2008empirical} and typical examples include estimation of causal parameters such as the conditional average treatment effect (CATE) and estimation of the conditional mean of an outcome when it is missing at random or censored at random. With particular focus on such coarsening problems, we aim to generally characterise the fundamental limits of nonparametric estimation and inference when estimating entire target functions using a generic machine learning method.

We do so by extending the standard notion of plug-in estimation and plug-in bias correction via influence functions (IFs), which are functional derivatives of statistical target parameters, from low-dimensional (see e.g. \cite{robins2017minimax}) to infinite-dimensional parameters. We propose a statistical learning framework we call `IF-learning', which is based on learning target functions from observed data using IFs and pseudo-outcomes motivated by IFs. We present two general learning algorithms that leverage these ideas -- the `IF-learner' and the `Group-IF-learner' -- which are best suited to learn in large and small sample regimes, respectively. After reviewing and developing necessary theory and algorithms, we finally return to our motivating example of inferring CATE and other causal parameters from experimental and observational data with a simulation study.

\subsection{Related work}
This paper builds on a straightforward question related to plug-in estimation. As we will discuss in detail in section \ref{review}, plug-in estimation exploits that target parameters of interest in applied statistics, such as (conditional) means can almost always be written as a functional of the underlying statistical model \citep{van2011targeted}. Since the statistical model is unknown, we can \textit{plug in} a data-driven estimate of this model to estimate a target parameter. Years of exciting research by particularly van der Laan and colleagues (see e.g. the exhaustive overviews in \citet{van2011targeted} and \citet{van2018targeted}) have established the idea that machine learning methods are naturally well-suited for this task. Unfortunately, plug-in estimation can lead to plug-in bias because the target functional is evaluated at the \textit{wrong} model, and the targeted maximum likelihood estimation (TMLE) literature (starting with \cite{van2006targeted}), as well as Robins and colleagues (e.g. in \citet{robins2017minimax}) have derived bias correction procedures, intuitively similar to Newton-Raphson steps, based on efficient scores and efficient influence functions (EIFs), that can be applied when target parameters are low-dimensional. 

Using TMLEs, recent work presented in \cite{van2018cv} has provided new ideas for infinite-dimensional pointwise plug-in bias adjustment based on kernel smoothers. While this allows the authors to make model-specific, problem-generic guarantees, the proposed approach is limited to the use of nonparametric  kernel smoothers and requires familiarity with the specifics of the TMLE framework. We are motivated by the intuitive appeal of constructing general `off-the-shelf' machine learning-based plug-in estimators for generic target functions, whose estimates can then be bias-corrected using influence functions. As we discuss later, we believe that the simple intuition of plug-in bias correction via influence functions -- a procedure similar to a Newton-Raphson Step -- could be of great intuitive appeal for practitioners (which is also the topic of the recent paper \cite{fisher2020visually}).

\begin{table}[!htb]
\begin{threeparttable}
\singlespacing
\centering
\setlength\tabcolsep{2pt}
\begin{tabular}{l|lllll}
\toprule
\midrule
Approach           & \begin{tabular}[c]{@{}l@{}}Infinite-\\ dim.\end{tabular} & Problem     & Model   & \begin{tabular}[c]{@{}l@{}}Based\\ on\end{tabular} & Inference                    \\
\toprule
\midrule
    \cite{robins2017minimax}        &                                                         & Generic     & Generic & EIF                                                & \checkmark                               \\
    \cite{chernozhukov2018double} & & Generic & Generic & NO &  \checkmark\\
    \cite{wager2018estimation}              & \checkmark                               & CATE        & Forest  &                                                  EE  & \checkmark\\
\cite{nie2017quasi}               & \checkmark                               & CATE        & Generic & Loss                                               &                                                                                      \\
\cite{lee2017doubly}                & \checkmark$^*$                               & CATE        &     Local-linear    & EIF                                                 &     \checkmark                          \\ 
\cite{fan2019estimation} & \checkmark$^*$ & CATE & Local-linear & NO  & \checkmark 
\\
\cite{kennedy2020optimal}            & \checkmark                               & CATE        & Generic & EIF                                                 &                               \\
\cite{Athey2019}              & \checkmark                               & Generic     & Forest  & EE                                                 &                                                       \checkmark \\
\cite{Chernozhukov2018}            & \checkmark$^*$                           & Generic$^*$ & Generic & EE                                                                                                        & \checkmark \\
\cite{foster2019orthogonal}             & \checkmark                               & Generic       &    Generic     & NO                                                  &                               \\
\cite{semenova2017estimation}           & \checkmark$^*$                           & Generic     & LS Series  & NO                                                  &     \checkmark                                                    \\
\cite{van2018cv}             & \checkmark                               & Generic     & Smoothing  & TMLE                                              & \checkmark \\
IF-learning (ours) & \checkmark                               & Generic     & Generic & EIF                                                                             & \checkmark \\
\midrule
\bottomrule
\end{tabular}
\begin{tablenotes}[para,flushleft]
\begin{footnotesize}
`Infinite-dim.' denotes whether the target parameter is infinite dimensional.  The column `inference' indicates whether the authors attempted to characterise conditions under which standard inference, e.g. based on a central limit theorem, is possible.  $^*$ denotes with additional qualifications/assumptions. CATE denotes Conditional Average Treatment effect. LS denotes Least Squares. We use the following abbreviations for concepts that methods are based on: Estimating Equations (EE), Efficient Influence Function (EIF), Loss-based (Loss), Neyman-Orthogonality (NO), Targeted Maximum Likelihood Estimation (TMLE). 
\end{footnotesize}
\end{tablenotes}
\end{threeparttable}
\caption{Conceptual overview of existing approaches for estimation and inference on infinite-dimensional structural target parameters and their capabilities}
\end{table}

Apart from using plug-in models and influence functions to construct estimators of target parameters, there exist multiple other strategies to do so. We broadly classify popular approaches into relying on  estimating equations, loss functions and/or Neyman-Orthogonality, and note that estimators can fall in multiple of these categories.  Based on some of these strategies, estimation of generic target functions has recently received increased attention in the fields of econometrics (e.g. \cite{Athey2019}, \cite{Chernozhukov2018}, \cite{foster2019orthogonal}, \cite{semenova2017estimation}) and biostatistics (e.g. \cite{van2018cv}), and, for the special case of CATE, particularly so in the causal inference communities (e.g. \cite{fan2019estimation}, \cite{kennedy2020optimal}, \cite{lee2017doubly},  \cite{nie2017quasi} and \cite{wager2018estimation}). Table 1 contains a schematic overview of the capabilities of approaches in related literature.  A more detailed literature review can be found in Appendix \ref{litreview} and throughout later sections.

\subsection{Outlook: IF-learning}
Complementary to some of the approaches in related literature, this paper was primarily motivated by a very simple question: \textit{How can we correct for plug-in bias that arises when we use a generic machine learning method for plug-in estimation of an infinite-dimensional structural target functional -- a function -- using influence functions?} The very unsatisfying technical answer to this question is: we cannot, because the (E)IFs of infinite-dimensional parameters do not exist. But the intuitive answer a machine learner would give, which lead to this paper, is a different one: we might not be able to correct for the plug-in bias \textit{exactly}, but we can \textit{approximately} do it. We propose two approaches to do so in this paper, leading to a statistical machine learning framework we call `IF-learning'.  Our main proposal is the `IF-learner', which, instead of using the true EIF (which does not exist) for plug-in bias adjustment, uses pseudo-outcome regression with outcomes motivated by the expression for the EIF, approximately adjusting entire functions for pointwise first-order plug-in bias. The `Group-IF-learner', on the other hand, does not learn the entire target function, but a group-wise approximation, for which EIFs do exist. 

The distinction between these two approaches -- learning a full function versus a coarser approximation -- is also motivated by a discussion in \cite{Chernozhukov2018}: even in arguably the simplest case, standard nonparametric regression, \cite{stone1980optimal}'s minimax convergence rate highlights that there exists no minimax consistent estimator in general if the dimension $d$ increases with sample size $n$, e.g. $d \geq \log(n)$\citep{Chernozhukov2018}. Further, the finite sample performance of nonparametric estimators deteriorates rapidly with increasing dimension $d$ for $n$ fixed. Inference on generic nonparametric function estimates is even more difficult -- in fact, adaptive confidence sets do not exist even for low dimensional nonparametric problems \citep{genovese2008adaptive}, as bias tends to dominate sampling error \citep{Chernozhukov2018}. This highlights that we have to prioritize between different goals in any case:  we cannot have the ability to estimate entire functions, the ability to perform inference on entire functions, the ability to consider high-dimensional data, and the ability to make no assumptions all simultaneously. 

Which of these abilities we are willing to give up may highly depend on the context. For example, when estimating treatment effects of new drugs in RCTs with small sample sizes, it might be much more important to retain the ability to perform valid statistical inference than to make individualised treatment recommendations, while the opposite may be true when building an in-hospital decision support system for daily practice trained on a large observational data set. The two settings -- access to experimental and observational data --  differ not only in the amount of available data (sample size), but also in the amount of information available to the statistician: In experimental settings, some features of underlying the statistical model are more likely to be \textit{known}, transforming fully nonparametric problems to proper semiparametric problems. For example, in RCTs the exact propensity score may be known, which is less likely to be the case in observational studies. The two learning algorithms we consider are each well-suited for one of these scenarios:
\begin{enumerate}
\item  The `IF-learner' for low-information, high sample size settings, in which we have access to large amounts of possibly low-quality data (e.g. unknown extent of selection on observables in treatment effect estimation), and rely on the assumption that the data-set is large enough such that finite sample bias is negligible. Using pseudo-outcome regression, it outputs an estimate of a full function that allows for individualised predictions -- for which statistical guarantees cannot be given beyond a minimax convergence rate, unless stricter assumptions are made. In this setting, we are bound to using data-sets that are not `too high-dimensional', i.e. $d < log(n)$. While it is motivated from the perspective of high-dimensional plug-in bias correction, this algorithm can also be seen as a generalization of \cite{kennedy2020optimal}'s CATE estimator to a much broader class of target parameters.
\item The `Group-IF-learner' for high-information, low sample size settings, in which we have access to small amounts of high-quality data, e.g. from a RCT in which propensity scores are \textit{known}. There, we cannot rely on asymptotic rate results only, however, due to information on e.g. selection mechanisms, we can obtain unbiased estimates for which standard inference is possible if we focus on a lower-dimensional approximation of our target. This algorithm outputs group-averaged target estimates with standard confidence intervals for data-adaptively determined heterogeneity groups, and is an adaptation of \cite{Chernozhukov2018}'s GATES algorithm.
\end{enumerate}

\paragraph{Contributions}
We aim to contribute to the growing literature applying traditional ideas from statistics to machine learning with the ultimate goal to facilitate the usage of flexible and data-adaptive methods in areas that are currently still restricted to the focus on simple population averages due to the  need to quantify estimation uncertainty and bias.  We hope to do so in this paper by highlighting the inherent usefulness of the concept of the EIF (even if it does not exist exactly) to characterise the ability of machine learning models (or data-adaptive nonparametric models more generally) to learn structural target functions of high interest in fields such as biostatistics and econometrics. We summarize our main contributions as follows: 
\begin{enumerate}
\item We propose `IF-learning', a statistical machine learning framework for learning structural target functions using `off-the-shelf' machine learning methods. Within this framework, we propose two first learning algorithms, the `IF-learner' and the `Group-IF-learner', suited to learn in high- and low-sample size settings, respectively. We hope that these very general algorithms can serve as templates for many problems of practical interest, in particular for problems with coarsening at random structure. Therefore, we also provide a `sklearn-style' python implementation of our algorithms\footnote{A preliminary version is available at \url{https://github.com/AliciaCurth/IF-learn}}.
\item While developing the `IF-learner', we propose to extend (i) the notion of plug-in bias removal and (ii) the notion of semiparametrically efficient estimation to entire target functions (for which EIFs do not exist) by performing pseudo-outcome regression using outcomes motivated by the population EIF, and show that learning target  functions in this manner requires no additional assumptions except for those associated with the regression method of choice. We also characterise the fundamental limits of this approach in fully nonparametric settings.
\item Through our theoretical analyses, we provide insights into the fundamental limits of learning structural target functions arising in a wide range of problems of interest in applied statistics --  namely problems for which the EIF of a population-averaged version of the target parameter does exist in analytic form. While \cite{stone1980optimal}'s nonparametric minimax rate remains the best that we can do to estimate \textit{any} function without further assumptions, our approach allows to re-characterise problems that are normally considered \textit{harder} than simple regression due to incomplete information and investigate how their difficulty compares to regression. A particularly interesting result of this investigation is that our debiased plug-in estimators in coarsening at random problems can achieve oracle rates whenever sufficient information about the coarsening mechanism or other parts of the problem is available, which means that incorporating domain knowledge into machine learning can substantially lower the bar on data requirements. We also illustrate our findings in a simulation study on treatment effect estimation.
\end{enumerate}

\paragraph{Structure of the paper}
We proceed as follows: In the following section we formalize the problem setting that we consider. Sections 3 and 4 are very theoretical in nature, and are intended to anchor our proposals in classical ideas from semiparametric statistics. Nonetheless, we give specific examples throughout to highlight practical implications. In Section \ref{review} we briefly review key concepts from semiparametric statistics that are necessary for general understanding of the origins of our proposals. Section 4 presents our key theoretical arguments, and highlights some first implications of combining our proposals with existing theoretical results. Sections 5 and 6 are more applied in nature: Section 5 constructs the proposed learning algorithms, and  Section \ref{sectionrct} contains a simulation study on estimation of different causal parameters. Section 7 concludes and highlights avenues of future research within our framework.

\section{Problem definition}
We discuss a very general problem setting covering many problems of interest in fields such as biostatistics and econometrics. By preserving this generality we can cover both well-studied standard problems, such as inference on a conditional mean on the one hand, and more intricate problems, such as treatment effect estimation, estimation under missing data and censoring, on the other. Throughout, we give examples to illustrate otherwise abstract concepts, with a focus on treatment effect estimation, the main motivating example of this paper.

Generally, we assume that we observe a sample of size $n$ of observations $O\sim \mathbbm{P}_0 \in \mathcal{P}$, where $O$ usually contains some $d$-dimensional covariate information $X \in \mathcal{X}$,  some outcome information $Y^* \in \mathcal{Y}$ and possibly other variables (see below). The statistical model $\mathbbm{P}_0$ is an element of $\mathcal{P}$, which -- without further restrictions -- contains all probability distributions over $O$ inducing a nonparametric estimation problem, and can be transformed to a proper semiparametric problem by imposing restrictions on its elements (usually guided by domain knowledge on data-generating processes). Our target parameter is an infinite-dimensional functional $\psi \equiv \psi(\mathbbm{P}_0)$ of the underlying statistical model $\mathbbm{P}_0$, most often a function of the form  $\psi(\cdot) \equiv \psi(\mathbbm{P}_0) (\cdot) : \mathcal{X} \rightarrow \mathbbm{R}$. 

In the class of \textit{coarsening at random} problems we are most interested in, we assume some additional structure (we follow the set-up in \cite{rubin2008empirical}). We assume that there exists a true underlying statistical model $F_0 \in \mathcal{F}$ generating data $Z=(Y^*, X)$. $Y^*$ denotes a (possibly multivariate) outcome variable, and $X \in \mathcal{X} \subset \mathbbm{R}^d$ contains characteristics associated with an observation.  Within this model, our interest lies in making inferences on a structural target function, $\psi: \mathcal{X} \rightarrow \mathbbm{R}$,  of the form
\begin{equation}\label{psiform}
\psi(x)=\mathbbm{E}_{{F}_0}[h(Y^*)|X=x]
\end{equation}
i.e. the conditional expectation of a function $h(\cdot)$ of the outcome variable $Y^*$.

In most scenarios of main interest, we do not fully observe $Y^*$, but instead obtain information only on a coarsened variable  $Y = \mathcal{C}(Y^*, C)$ . Here, $C$ is a coarsening variable determining what we observe and the mapping $\mathcal{C}$ is deterministic. Such coarsening mechanisms are prevalent in settings involving, for example, counterfactual inference, missing data or censoring. In high information settings (e.g. learning from experimental data), we assume that the stochastic coarsening mechanism $C|X \sim G_0(\cdot)$ is known or easily estimable, while in low information settings (e.g. learning from observational data) we leave it completely unspecified.

 Thus, instead of $\{Z_i\}^{n}_{i=1} = \{(Y^*_i, X_i)\}^{n}_{i=1}$ we observe a random sample $\{O_i\}^{n}_{i=1} = \{(Y_i, X_i, C_i)\}^{n}_{i=1}$, where $Y_i$ is the observed coarsened variable (and hence could also contain \textit{no} information), drawn i.i.d. from a probability distribution $O \sim \mathbbm{P}_0(\cdot) = \mathbbm{P}_{F_0, G_0}(\cdot)$, determined by the statistical model $F_0$ and the coarsening mechanism $G_0$. While we do not observe $Y^*$, we make the assumption that we can nonetheless \textit{identify} $\psi(x)$ from the observed data. We thus consider problems in which we can construct signals $\tilde{Y}=f(O)$, such that $\mathbbm{E}_{\mathbbm{P}_0}[\tilde{Y}|X=x]=\psi(x)$ (at least asymptotically). Fortunately, as discussed above, once identification of the target parameter is guaranteed, all remaining problems are statistical in nature and concern only estimation of  the target.
 
Below, we give some examples to illustrate typical target parameters with such structure. We begin with standard regression and classification problems, because they are trivial examples of coarsening at random problems (\textit{no} coarsening). Further, the main motivating example in this paper, which is also discussed in more depth in section \ref{sectionrct}, is discovering evidence for heterogeneous treatment effects from experimental and observational studies. We maintain this example throughout as it is not only well-known and of high interest in many communities but also allows us to highlight the value of having knowledge of the coarsening mechanism (here: treatment assignment mechanism). Finally, we give two more examples that we will not treat in detail but present to illustrate the breadth of problems this framework can cover.

\begin{example}[Conditional means in standard regression and binary classification]
The most simple example for a target function is the conditional mean $\psi(x)=\mathbbm{E}_{\mathbbm{P}_0}[Y|X=x]$ in nonparametric regression, making no assumptions on the data-generating process. Here, $O_i=Z_i=(Y_i, X_i)$, since there is no coarsening, i.e. $h(Y^*)=Y \text{ and } \tilde{Y}=Y$. The same holds for the conditional success probability  $\psi(x) = \mathbbm{P}_0(Y=1|X=x) = \mathbbm{E}_{\mathbbm{P}_0}[\mathbbm{1}\{Y=1\}|X=x]$, the target parameter in binary classification.
\end{example}

\begin{example}[Motivating example: Conditional average treatment effects]
Given a binary treatment $W \in \{0, 1\}$, assigned according to propensity score $\pi(x) = \mathbbm{P}_0(W=1| X=x)$ -- which is assumed known in experimental settings -- we are interested in an individualised treatment effect: the difference between the potential outcomes $Y_i(0)$ if individual $i$ does not receive treatment ($W_i=0$) and $Y_i(1)$ if treatment is administered ($W_i=1$). If we had access to both potential outcomes, then the individual treatment effect $Y_i(1) - Y_i(0)$ would be a natural outcome of interest. However, by the \textit{fundamental problem of causal inference}, we only ever observe one of the two potential outcomes. Therefore, the majority of existing literature focuses on the conditional average treatment effect (CATE), $\tau(x) = \mathbbm{E}_{\mathbbm{P}_0}[Y(1) - Y(0) |X =x]$, the expected treatment effect for an individual with covariate values $X=x$. 

In this case, $Y^* = (Y(0), Y(1))$, $h(Y^*)= Y(1) - Y(0)$ and $\psi(x)=\tau(x)$. We observe only $Y= W Y{(1)} + (1- W) Y{(0)}$, the potential outcome associated with the received treatment, so the treatment indicator $W$ acts as the coarsening variable $C$.  Under some additional assumptions (discussed in section \ref{sectionrct}), we can construct an unbiased Horvitz-Thompson-type signal $\tilde{Y}=(\frac{W}{\pi(X)} - \frac{(1-W)}{1-\pi(X)})Y$ \citep{horvitz1952generalization} for the estimand of interest from the observed data -- which, as we will illustrate later, is not the best way of estimating the CATE.
\end{example}

\begin{example}[Other causal parameters]
While much of the causal inference literature in machine learning concentrates on (C)ATE, it is not actually the only parameter of practical interest \citep{rubin2008empirical}. For example, if we let $\mu_0(x)$ and $\mu_1(x)$ denote the expected values of the potential outcomes, and the potential outcomes are binary, i.e. $Y(w) \in \{0, 1\}$, then the relative risk $f_{RR}(\mu_0(x), \mu_1(x))=\frac{\mu_1(x)}{\mu_0(x)}$ or the odds ratio $f_{OR}(\mu_0(x), \mu_1(x))= \frac{\mu_1(x)}{1-\mu_1(x)}/ \frac{\mu_0(x)}{1-\mu_0(x)}$ are often the parameters of natural interest in Randomized Control Trials (RCTs). In section \ref{sectionrct} we consider estimating such parameters using the transformation $f(\psi(x))$. 

\end{example}
\begin{example}[Conditional expectations when outcomes are missing at random]
Another well-known example of coarsening at random is that of missing outcome data (e.g. in a clinical trial). Here, we observe $O_i = (A_i Y_i, A_i, X_i)$, where $A_i\in \{0,1\}$ is a missingness indicator and $Z_i=(Y_i, X_i)$. If we know that data is missing at random, i.e. the probability of missingness $\pi(x) = \mathbbm{P}_0(A=1|X=x)$ is determined only by covariate information, then $\psi(x) = \mathbbm{E}_{\mathbbm{P}_0}[Y| A=1, X=x]$, the expected outcome of those subjects with missing data, would be the functional of interest. Clearly, this problem has a structure very similar to that of treatment effect estimation, and also admits a Horvitz-Thompson-type signal $\tilde{Y}$. A similar argument holds also for the expected outcome value of censored outcomes.
\end{example}

We approach the problem from a semiparametric statistics viewpoint, because we wish to make little to no assumptions on the data generating distribution $\mathbbm{P}_0$, which in reality could be arbitrarily complex. We wish to impose (parametric) restrictions only when we are sure that they are a feature of the underlying problem -- which could be the case if domain knowledge on the problem at hand is available. In this paper, we consider only fully nonparametric settings.  Our interest in $\mathbbm{P}_0$ is purely motivated by the target functional 
$\psi = \psi(\mathbbm{P}_0)$ -- all other components of $\mathbbm{P}_0$ are \textit{nuisance parameters} to us. Some of these nuisance parameters are more important for estimation of $\psi$ than others, sometimes we refer to those parameters that are needed for (efficient) inference on $\psi$ as $Q \equiv Q(\mathbbm{P}_0)$. For example,  $Q$ often includes parameters such as conditional mean functions and coarsening probabilities, while the distribution of zero-mean error terms can often be completely ignored. 

In classical semiparametric statistics, the target parameter is typically very low-dimensional -- e.g. a one-dimensional population average such as the average treatment effect. In stark contrast to this, recall that our target parameter is itself a function $\psi: \mathcal{X} \rightarrow \mathbbm{R}$ and hence infinite-dimensional. We assume continuous covariates throughout, yet our approach is also of practical use if covariates are discrete but high dimensional.  In the settings that we consider, $\Psi = \mathbbm{E}_{X \sim \mathbbm{P}_0}[\psi(X)]$, the population average of our target parameters, has been well-studied. Estimation of and inference on $\psi$ itself is still very much a topic of active research. 

In this paper, we aim to characterise the limits of estimating general structural target functions $\psi(x)$ using generic nonparametric machine learning methods with minimax performance guarantees, and use our findings to construct learning algorithms that are of practical value to applied researchers. To enable consistent estimation of high-dimensional functionals in settings where not all data is observed, we first need to develop simple strategies for plug-in bias removal based on pseudo-outcome regression. This is the main focus in the theoretical part of this paper (Sections \ref{review} and \ref{sectheor}). We also rely on semiparametric efficiency theory to construct estimators that are approximately pointwise efficient. For both, we make heavy use of the efficient influence functions of target parameters, a concept that we will define in the following section.

\section{Theoretical background}\label{review}
We begin by briefly reviewing key concepts and strategies used in semiparametric statistics to construct estimates of low-dimensional target parameters. In the remainder of this paper, we will build on and generalise these ideas to high-dimensional target parameters. We refer the reader to \cite{van2014higher}, \cite{kennedy2016semiparametric}, \cite{kennedypres}  and \cite{fisher2020visually} for excellent introductions to influence functions and related concepts, and how these arise in semiparametric statistics. For more exhaustive treatment of existing approaches in semiparametric statistics, tailored mainly to applications in biostatistics,  we refer to \cite{tsiatis2007semiparametric}, \cite{van2003unified} and \cite{van2011targeted}.

The concepts of main interest in this paper are the influence functions (IFs) and, in particular, the efficient influence function (EIF) of a target parameter. Influence functions arise naturally in multiple areas of statistics, most prominently in the area of robust statistics, where the \textit{influence function of an estimator} originally measures the robustness of an estimator to outliers \citep{hampel1980robust}. This is \textit{not} the use case of interest for influence functions in this paper. Instead, we build on \textit{influence functions of target parameters} as they are used in semiparametric statistics, namely in the context of (i) plug-in estimation and plug-in bias correction and (ii) characterisation of asymptotically efficient estimators. We will first introduce influence functions generally, and then discuss these two aspects in turn. As is standard in most of the literature, unless stated otherwise, we will assume in this section the one-dimensional problem in which our target is a population average $\Psi \equiv \Psi(\mathbbm{P}_0) = \mathbbm{E}_{X \sim \mathbbm{P}_0}[\psi(X)]$.

\subsection{Influence functions}
It is beyond the scope of this paper to explain how IFs were originally derived using tangent spaces, for a comprehensive introduction we refer to \cite{tsiatis2007semiparametric}. Instead, we give an intuitive introduction relating IFs of statistical parameters to derivatives of analytical functions, inspired by the discussion in \cite{fisher2020visually}. 

For a distribution $\mathbbm{P}$ with density $p$ we can define a distribution  $\mathbbm{P}_\epsilon$ with density $p_\epsilon$ given by
\begin{equation}
p_\epsilon(o) = (1 - \epsilon)p(o) + \epsilon \tilde{p}(o)
\end{equation}
with $\tilde{p}(o)$ the density of another distribution $\tilde{\mathbbm{P}}$  and $\epsilon$ small. Since our target parameter $\Psi(\mathbbm{P})$ is a functional of a distribution, we can use $\mathbbm{P}_\epsilon$  to evaluate the sensitivity of our target $\Psi$ to small changes along a path $\{\mathbbm{P}_\epsilon\}_{\epsilon \in [0,1]}$ where $\epsilon$ at the end points $0$ and $1$ reduces   $\mathbbm{P}_\epsilon$ to the original distributions $\mathbbm{P}$ and  $\tilde{\mathbbm{P}}$, respectively. 

Now assume we have a given \textit{plug-in} model $\tilde{\mathbbm{P}}$ for  $\mathbbm{P}$ and we would like to use the intuition of the path along small changes $\epsilon$ defined above to correct the bias induced by evaluating $\Psi$ at the \textit{wrong} distribution $\tilde{\mathbbm{P}}$ instead of $\mathbbm{P}$. If $\Psi$ was an analytical function, we would do so by using a Newton-Raphson step or a first-order Taylor-expansion. It turns out that plug-in bias correction for functionals can be handled using the exact same idea, and we only need a functional generalization of a derivative to essentially use the same procedures. This \textit{is} the IF. 

\begin{definition}[Influence Function (adapted from \cite{fisher2020visually})]
For a given functional $\Psi$, an influence function for $\Psi$ is any function  $\dot{\Psi}$  satisfying
\begin{equation}
\frac{\partial \Psi(\mathbbm{P} + \epsilon(\tilde{\mathbbm{P}}-\mathbbm{P}))}{\partial \epsilon}\Bigr|_{\epsilon=0} = \int \dot{\Psi}(O, \mathbbm{P}) (d\tilde{\mathbbm{P}} - d \mathbbm{P})
\end{equation}
and  
\begin{equation}\label{ifunbiased}
\int \dot{\Psi}(O, \mathbbm{P}) d \mathbbm{P} = \int (D_{\Psi, \mathbbm{P}}(O) - \Psi) d \mathbbm{P} = 0
\end{equation}
\end{definition}
Property (\ref{ifunbiased}) implies that the \textit{uncentered} influence function $D_{\Psi, \mathbbm{P}}(O)$ is unbiased for $\Psi$, i.e. that $\int D_{\Psi, \mathbbm{P}}(O)d \mathbbm{P}=\Psi$. This property is the basis for the construction of regular and asymptotically linear (RAL) estimators  based on influence functions in semiparametric statistics (see next section). $\dot{\Psi}$ need not be unique -- if restrictions are placed on the underlying model the number of influence functions is infinite. In fully nonparametric problems, the influence function is unique (if it exists) and is referred to as the \textit{efficient influence function} because it can be used to construct the most efficient unbiased semiparametric RAL estimator of $\Psi$. Because of this property, we will \textit{only} consider efficient IFs (EIFs) in this paper, and sometimes drop the term efficient for brevity. Finally, by extension of the intuition of derivatives of analytical functions, many of the standard rules of calculus such as the chain- and product rules hold for influence calculus, implying that influence functions for some seemingly complex parameters can be built up from simple building blocks \citep{kennedypres}. The simulation study contained in section \ref{sectionrct} will highlight one such example.

Unfortunately, infinite-dimensional parameters in nonparametric models, e.g. parameters that arise as functions of \textit{continuous} inputs, are not pathwise differentiable and hence their IFs do not exist, which is tied to the intuition that singletons are not measurable \citep{van2018cv}. Therefore, IFs are not defined for our problems of interest -- a topic we revisit in section \ref{maincontrib}. Before we move on, we give two important examples of influence functions that we will use throughout this paper.

\begin{example}[EIF of the mean] In generic (unrestricted) nonparametric regression ($\psi(x)= \mathbbm{E}_\mathbbm{P}[Y|X=x]$), the efficient influence function of $\Psi$ is ${\dot{\Psi}(O; \mathbbm{P}) = Y - \Psi}$. That is $D_{\Psi, \mathbbm{P}}(O) = Y$ is independent of nuisance parameters. 
\end{example}
\begin{example}[EIF of the average treatment effect] For the average treatment effect, \\ $\Psi=\mathbbm{E}_{X \sim \mathbbm{P}}[\tau(X)]$, the uncentered efficient influence function is given by 
\begin{equation*}
D_{\Psi, \mathbbm{P}}(O)=  \left(\frac{W}{\pi(X)}- \frac{(1-W)}{1-\pi(X)}\right) Y + \left[\left(1 - \frac{W}{\pi(X)}\right) \mu_1(x)-\left(1 - \frac{1-W}{1-\pi(X)}\right)\mu_0(X)\right]
\end{equation*}
with $\mu_w(x) = \mathbbm{E}_{\mathbbm{P}}[Y(w)|X=x)] = \mathbbm{E}_{\mathbbm{P}}[Y|W=w, X=x] \text{, } w \in \{0,1\}$. This has the same form as the well-known augmented inverse propensity weighted (AIPW) estimator \citep{robins1995semiparametric}.
\end{example}

\subsection{Plug-in estimation and correcting for plug-in bias}\label{pluglow}
IFs naturally arise in the context of bias-correction in \textit{plug-in} estimation.  Plug-in estimation exploits that $\Psi$ is a \textit{functional} mapping a statistical model $\mathbbm{P}\in \mathcal{P}$ to $\Psi(\mathbbm{P}) \in \mathbbm{R}$. If we can construct an estimator $\hat{\mathbbm{P}}$ of $\mathbbm{P}_0$ from a sample $\mathcal{D}=\{O_i\}^n_{i=1} \sim \mathbbm{P}_0$, we can estimate $\Psi(\mathbbm{P}_0)$ by $\Psi(\hat{\mathbbm{P}})$. Unfortunately, such plug-in estimators often inherit considerable first-order bias from the nonparametric estimators $\hat{Q}$ contained in $\hat{\mathbbm{P}}$, i.e. the estimators for high-dimensional nuisance parameters $Q(\mathbbm{P}_0)$ needed for estimation of $\Psi$. Since $\Psi(\mathbbm{P}_0)$ and $\Psi(\hat{\mathbbm{P}})$ are evaluations of \textit{the same functional} at different inputs, we can conceptualise this bias using a generalisation of the Taylor expansion to functionals, the von Mises expansion of $\Psi(\mathbbm{P})$:
\begin{equation}\label{vanmises}
\Psi(\mathbbm{\hat{P}}) - \Psi(\mathbbm{P}_0) = \int \dot{\Psi}(O;\hat{\mathbbm{P}})d(\mathbbm{\hat{P}} - \mathbbm{P}_0) + R_2(\mathbbm{\hat{P}}, \mathbbm{P}_0)  = - \int \dot{\Psi}(O; \hat{\mathbbm{P}})d \mathbbm{P}_0 + R_2(\mathbbm{\hat{P}}, \mathbbm{P}_0)
\end{equation}
where $R_2$ is a second order remainder. If it exists, the EIF discussed in the previous section satisfies (\ref{vanmises}), so we will generally consider $\dot{\Psi}(O;\mathbbm{P}) = D_{\Psi, \mathbbm{P}}(O) - \Psi$. Equation (\ref{vanmises}) suggests that we can approximately remove the first-order bias by using the  following adjusted estimator:
\begin{equation}\label{plugin}
\tilde{\Psi}(\hat{\mathbbm{P}}) = \Psi(\mathbbm{\hat{P}}) + \frac{1}{n} \sum^n_{i=1}\dot{\Psi}(O_i;\hat{\mathbbm{P}}) = \frac{1}{n} \sum^n_{i=1} D_{\Psi, \hat{\mathbbm{P}}}(O_i)
\end{equation}
where the second equation follows from $\dot{\Psi}(O, \hat{\mathbbm{P}}) = D_{\Psi, \hat{\mathbbm{P}}}(O) - \Psi(\hat{\mathbbm{P}})$.

\begin{remark}[Higher-order bias adjustment]
Here, we consider only only first order bias adjustment. If the second-order remainder $R_2$ decays too slowly, we can do better in theory by also using higher-order influence functions \citep{robins2008higher} which generalise higher-order Taylor expansions. In practice, higher-order influence functions, while correcting for higher-order bias, can also add substantial variance in finite samples \citep{van2014higher}.
\end{remark}

From equation (\ref{vanmises}) and (\ref{plugin}) it becomes obvious that we can write the plug-in bias-adjusted estimator $\tilde{\Psi}$ as a regular asymptotically linear (RAL) estimator\footnote{While the sum in the RHS of (\ref{ralexpression}) illustrates why such estimators are referred to as asymptotically linear, it is beyond the scope of this paper to discuss regularity of estimators. We refer the reader to e.g. \cite{tsiatis2007semiparametric} for an introduction to regular estimators.} obtained from a sample of size $n$, i.e.
\begin{equation}\label{ralexpression}
\sqrt{n}(\tilde{\Psi}(\hat{\mathbbm{P}}) - \Psi(\mathbbm{P}_0) )= \frac{1}{\sqrt{n}}\sum^n_{i=1} (D_{\Psi, \hat{\mathbbm{P}}}(O_i) - \Psi(\mathbbm{P}_0)) + o_{\mathbbm{P}_0}(1)
\end{equation}
if (i) $\sqrt{n}R_2(\hat{\mathbbm{P}}, \mathbbm{P}_0)$ is negligible ($o_{\mathbbm{P}_0}(1)$) and (ii) the bias induced in $\Psi$ by approximating $\mathbbm{P}_0$ with $\hat{\mathbbm{P}}$ is negligible. (i) can be achieved if $O_{\mathbbm{P}_0}(R_2(\hat{\mathbbm{P}}, \mathbbm{P}_0))= n^{-1/2}$,
whereas (ii) needs that $\tilde{\Psi}(\hat{\mathbbm{P}})$ is fit on an independent sample to that used for estimating $\hat{\mathbbm{P}}$ (unless we are restricting our attention to model classes $\mathcal{P}$ that are Donsker) \textit{and} that $\dot{\Psi}(\hat{\mathbbm{P}})$ is consistent for  $\dot{\Psi}({\mathbbm{P}_0})$ in l2-norm \citep{kennedypres}. If these conditions hold, then $\tilde{\Psi}(\hat{\mathbbm{P}})$ is asymptotically unbiased and $\sqrt{n}(\tilde{\Psi}(\hat{\mathbbm{P}}) - \Psi({\mathbbm{P}_0}))$ is  asymptotically normal with mean zero and minimum attainable variance (the semiparametric equivalent of the parametric Cramer-Rao Lower bound).

\section{Efficient estimation of structural target functions}\label{sectheor}
In this section, we theoretically motivate and analyse group-wise estimation and pseudo-outcome regression for approximate plug-in bias removal across target functions, leading to our main learning algorithms presented in Section 5. 

\subsection{Efficient group-averaged target estimators}\label{maincontrib}
This paper is based on the following, very intuitive idea: Although the EIF does not exist for infinite-dimensional target parameters \citep{van2018cv}, we can construct a sequence of pathwise differentiable target parameters approaching $\psi(x)$ arbitrarily closely. This idea leads both to the `Group-IF-learner' -- which estimates only a group-wise approximation to $\psi(x)$ using pathwise differentiable target parameters -- and the `IF-learner', which considers the nondifferentiable limit.  For ease of exposition, we consider only target parameters that can be written as (linear functions of) conditional outcome means here, and consider an example with nonlinear target parameters, e.g. risk ratios, in section \ref{sectionrct}.  Throughout, we will make heavy use of the \textit{uncentered} population EIF $D_{\Psi, \mathbbm{P}}$. 

To construct the `Group-IF-learner' in section 5.2 and to motivate the pseudo-outcome regression set-up we discuss next, we begin by characterising the EIF and efficient estimators of partition-based group-averaged targets. We loosely define a (fixed) partition of the input space $\mathcal{X}$ as $\pi =\{A^\pi_k \subset \mathcal{X}: A^\pi_k \cap A^\pi_j = \emptyset \text{ } \forall k \neq j\text{, } \cup^{K^\pi}_{k=1} A^\pi_k = \mathcal{X}\} \text{ for some } K^\pi  \geq 1$ and, for ease of exposition, let $\pi$ be any type of connected partition, e.g. those consisting of voronoi cells.

We first characterise the efficient estimator of the target within the $k^{th}$ cell of our partition,  $\Psi^\pi_k = \mathbbm{E}_{\mathbbm{P}_0}[\psi(x)|x \in A^\pi_k]$.

\begin{proposition}[Efficient estimator for group-averaged target parameters]\label{partprop} 
Let $\psi(\mathbbm{P})(x)$ be a target parameter that is a linear function of conditional mean(s). If $\hat{\Psi}=\frac{1}{n}\sum^n_{i=1} D_{\Psi, \hat{\mathbbm{P}}}(O_i)$ is the asymptotically efficient estimator for the population mean $\Psi(\mathbbm{P}) = \mathbbm{E}_{\mathbbm{P}}[\psi(X)]$ of a target parameter $\psi(x)$,  based on the efficient influence function of the form $D_{\Psi, \mathbbm{P}}(O) - \Psi(\mathbbm{P})$, then $\hat{\Psi}^\pi_k=\frac{1}{|\{i: X_i \in A^\pi_k\}|}\sum_{i: X_i \in A^\pi_k} D_{\Psi, \hat{\mathbbm{P}}}(O_i)$ is the efficient estimator for $\Psi^\pi_k$.
\end{proposition}
\begin{proof}
This result is intuitive, and the proof is straightforward and stated in Appendix \ref{proofprop1}. 
\end{proof}

Based on this result, we can construct a target parameter that is arbitrarily close to $\psi(x)$ but for which both the EIF and resulting efficient estimator do exist. For any $x \in \mathcal{X}$, if we define a sequence of cells $A_x^\epsilon=[x-\epsilon, x+\epsilon]$, indexed by $\epsilon>0$, the EIF of the target parameter $\Psi_{x, \epsilon} = \mathbbm{E}_{\mathbbm{P}_0}[\psi(x)|x \in A^\epsilon_x]$ does exist for all $\epsilon >0$, resulting in efficient estimator $\hat{\Psi}_{x, \epsilon}= \frac{1}{|\{i: X_i \in  A^\epsilon_x\}|}\sum_{i: X_i \in  A_x^\epsilon} D_{\Psi,\hat{\mathbbm{P}}}(O_i)$. 

At $\epsilon=0$, the EIF is undefined, however, the efficient estimator $\hat{\Psi}_{x, \epsilon}$ does have a limit if $\exists i \in \{1, \ldots, n\}: X_i = x$, namely $\hat{\Psi}_{x, 0} = \frac{1}{|\{i: X_i = x\}|}\sum_{i: X_i = x} D_{\Psi, \hat{\mathbbm{P}}}(O_i)$ (and is undefined otherwise). Thus, if there is a value $X_i \in \{1, \ldots, n\}: X_i = x$ in our data-set, and if $\mathbbm{P}_0$ is known, then we can construct an unbiased estimator using the expression for the population EIF, namely $\mathbbm{E}_{\mathbbm{P}_0}[\hat{\Psi}_{x, 0}] = \mathbbm{E}_{\mathbbm{P}_0}[D_{\Psi, \mathbbm{P}_0}(O)| X=x]= \psi(x)$. If we had oracle knowledge of $\mathbbm{P}_0$, we could use this idea to construct $\psi(x)$ by pointwise estimation. However, the pointwise efficient estimator which estimates $\mathbbm{E}[D_{\Psi, \mathbbm{P}_0}(O)|X=x]$ by the empirical average $\hat{\Psi}_{x} = \frac{1}{|\{i: X_i = x\}|}\sum_{i: X_i = x} D_{\Psi, {\mathbbm{P}_0}}(O_i)$ over all observations with $X_i=x$ clearly is not feasible if $X$ is continuous. This can even be the case if Definition 1 holds -- namely when $X$ is discrete but high dimensional. 
\subsection{Approximate plug-in bias removal via pseudo-outcome regression}
Even though the direct estimator is infeasible, we \textit{can} leverage ideas from nonparametric statistics and machine learning to \textit{learn} the pointwise expected value from data using pseudo-outcome regression: If we let $D_{\psi,i} \equiv D_{\psi, {\mathbbm{P}_0}}(O_i)$ denote the pseudo-outcome which we obtain by evaluating the analytical form $D_{\Psi, {\mathbbm{P}_0}}(\cdot)$ of the population EIF at each point $i$ in our data-set, then we could estimate $\psi(x)$ by regressing $\{D_{\psi,i}\}^n_{i=1}$ on $\{X_i\}^n_{i=1}$.  When our target parameter is a nonlinear function of a conditional expectation, we can still construct an unbiased pseudo-outcome $D_{\psi}$, however, this will not take the exact form of the uncentered population EIF (see section \ref{sectionrct} for an example). Below, we formalize that in an oracle setting ($\mathbbm{P}_0$ known) and given identifiability of $\Psi$ using $D_{\Psi, \mathbbm{P}_0}$, this approach relies only on the assumptions associated with the regression method of choice.

\begin{proposition}[Learning target functions via oracle pseudo-outcome regressions]
Given identifiability of $\Psi$, access to an oracle with knowledge of the true $\mathbbm{P}_0$ and a learning algorithm $\mathcal{A}$, we can learn the expected value of  $D_{\psi, {\mathbbm{P}_0}}(O)$ from data at the minimax rate associated with $\mathcal{A}$ under no additional assumptions than the standard  assumptions associated with simple regression using $\mathcal{A}$.

If $\mathcal{A}$ is a generic nonparametric regression estimator, and   $\psi(x)$ is $p$-smooth, then, under the standard regularity conditions (sketched in appendix \ref{condthm1}), the minimax convergence rate of this oracle regression is \cite{stone1980optimal}'s minimax rate $n^{-p/(2p+d)}$. 
\end{proposition}
\begin{proof}
Since  $\mathbbm{E}_{\mathbbm{P}_0}[D_\psi| X= x] = \mathbbm{E}_{\mathbbm{P}_0}[D_{\psi, {\mathbbm{P}_0}}(O)| X=x]= \psi(x)$, and we are in an oracle setting, i.e. have knowledge of the model $\mathbbm{P}_0$, we have an observation $D_{\psi, {\mathbbm{P}_0}}(O_i)$ whenever we have an observation $O_i$ for which $X_i=x$. Therefore, we can think of $D_{\psi} \equiv D_{\psi, {\mathbbm{P}_0}}(O)$ as a pseudo-outcome in nonparametric regression with conditional mean $\psi(x)$ and some random error. Thus, if we assume that $\psi(x)$ is sufficiently regular to be estimable from the data using regression methods $\mathcal{A}$, then it must be possible to do so using the `canonical' pseudo-outcome $D_{\psi}$.

For generic nonparametric regression, the assumptions that need to be placed on the underlying statistical model $\mathbbm{P}_0$ include some regularity conditions on the fixed function $\psi(x)$, and the main requirement for the random pseudo-outcomes $D_{\psi}$ are bounded first and second moments. For a sketch of the full conditions as originally characterised by \cite{stone1980optimal}, refer to appendix \ref{condthm1}.
\end{proof}

As this estimator is also infeasible due to a lack of oracle knowledge of $\mathbbm{P}_0$ in practice, we propose a two-stage estimator which first estimates a plug-in model $\hat{\mathbbm{P}}$ from the data, and then (analogously to the low-dimensional setting discussed in section 3.2) removes the plug-in bias in $\psi(\hat{\mathbbm{P}})$ through regression of $D_\psi$ on $X$. As we will show using a result of \cite{kennedy2020optimal} in the next section, to be able to make guarantees on this estimator, we will have to perform the two steps on two separate samples, which is intuitively obvious due to the inherent risk of overfitting and follows by direct extension from the low-dimensional case discussed in section 3.3.

\subsubsection{Theoretical analysis}\label{pracimp}
The first natural questions that arise in the pseudo-outcome regression set-up we constructed above are (i) how the rate of convergence of this estimator compares to the oracle rate ($\mathbbm{P}_0$ known), and (ii) how it compares to the original plug-in estimator.

The key to answering these questions is given in the recent work of \cite{kennedy2020optimal}, providing a general result to bound the error of pseudo-outcome regression with imputed or estimated components, which gives us the ability to bound the error of our plug-in estimates when estimated using two independent samples of size $n$. For completeness we state \cite{kennedy2020optimal}'s theorem in appendix \ref{kennedythm}. To be able to bound the error in plug-in pseudo-outcome regression, we need a mild assumption on the second-stage regression model (see appendix \ref{kennedythm}) necessary to ensure stability of the second stage regression \citep{kennedy2020optimal}. Further, we make the following very weak assumptions on the set-up for plug-in pseudo-outcome regression for an infinite-dimensional target parameter $\psi(x)$.
\begin{assumption}[Set-up of plug-in-based, bias-corrected pseudo-outcome regression]\label{asssetup}
We assume the following set-up for our plug-in-based, bias-corrected pseudo-outcome regression:
\begin{itemize}
\item We use sample splitting or cross fitting for the plug-in regression. In a sample-splitting set-up, let  $O_0^n = ({O_{01}, \ldots, O_{0n}})$ and $O_1^n = ({O_{11}, \ldots, O_{1n}})$ denote the two independent samples used for first- and second-stage regression, respectively.
\item We have pseudo-outcomes $\{D_{\psi, \hat{\mathbbm{P}}}(O_i)\}_{i \in O_1^n}$ based on some first stage model $\hat{\mathbbm{P}}$ estimated using $O_0^n$
\item We have a plug-in bias-corrected estimator $\hat{\psi}(x)$ based on the plug-in pseudo-outcomes $\{D_{\psi, \hat{\mathbbm{P}}}(O_i)\}_{i \in O_1^n}$ which we use as a target in the second stage pseudo-outcome regression
\item There is an oracle equivalent $\tilde{\psi}(x)$ of $\hat{\psi}(x)$ from an infeasible regression of the true  $\{D_{\psi, {\mathbbm{P}}}(O_i)\}_{i \in O_1^n}$ on $\{X_i\}_{i \in O_1^n}$
\item $\psi(x)$ is sufficiently regular so that Proposition 2 holds and we can learn the expected value of $D_{\psi, {\mathbbm{P}}}(O_i)$ from data
\end{itemize}
\end{assumption}
These two assumptions, and \cite{kennedy2020optimal}'s theorem immediately give us the following result:
\begin{corollary}[Error bound for plug-in-based, bias-corrected pseudo-outcome regression estimators of structural target functions]\label{errorbound} 
Under assumptions \ref{kennedyass} and \ref{asssetup} we have the following result for pseudo-outcome regression of $D_{\Psi, \hat{\mathbbm{P}}}(O)$ on $X$
\begin{equation}\label{errorpsi}
\mathbbm{E}\left[\{\hat{\psi}(x) - \psi(x)\}^2\right] \lesssim \mathbbm{E}\left[\{\tilde{\psi}(x) - {\psi}(x)\}^2\right] + \mathbbm{E}\left[\{\mathbbm{E}[D_{\psi, \hat{\mathbbm{P}}}(O)|X=x, O_0^n] - \psi(x)\}^2\right]
\end{equation}
\end{corollary}
\begin{proof}
This follows directly from the assumptions and \cite{kennedy2020optimal}'s theorem.  
\end{proof}
Because the shape of the second-order remainder term on the right hand side depends on the specific target parameter of interest, it is not immediately clear what this result implies in general. Therefore, it can be instructive to consider specific examples. First, since plug-in bias correction should have no effect when estimating the mean $\mu(x)=\mathbbm{E}_{\mathbbm{P}}[Y|X=x]$ in simple nonparametric regression, as we have shown in section \ref{review} that the uncentered EIF ($D_{\mu, \mathbbm{P}}=Y$) is independent of $\mathbbm{P}$ by construction, it is reassuring that the second term here will be zero since there is no second-order remainder in nonparametric regression \citep{kennedypres}. Thus, as expected, nonparametric regression attains the oracle rate. Further, in the class of coarsening at random problems, this remainder has the familiar doubly robust form, meaning that the plug-in bias removal step loosens the requirements on convergence also in the higher-dimensional case, which is a familiar notion extending directly from the standard, low-dimensional case (refer to appendix \ref{CATEImpl} for an example of how this works mathematically in CATE estimation).

Corollary \ref{errorbound} immediately leads to the following result on convergence rates: 
\begin{corollary}[Minimax convergence rates for generic nonparametric estimation of infinite-dimensional target parameters using plug-in-based, bias-adjusted pseudo-outcome regression]\label{ratethm} 
Suppose the assumptions of Corollary \ref{errorbound} hold. Further, assume that the target function $\psi$ is $p$-smooth and can be estimated with minimax rate $n^{-\frac{p}{2p+d}}$, where $x \in \mathbbm{R}^d$. 
Then 
\begin{equation}
\mathbbm{E}\left[\{\hat{\psi}(x) - \psi(x)\}^2\right] \lesssim \max(n^{-\frac{2p}{2p+d}}, O_{\mathbbm{P}_0}(R_2))
\end{equation}
where $R_2=\mathbbm{E}\left[\{\mathbbm{E}[D_{\psi, \hat{\mathbbm{P}}}|X=x, O_0^n] - \psi(x)\}^2\right]$.\\
In particular, if the order of the oracle term dominates the remainder term, we can achieve oracle rates with the proposed plug-in estimator, i.e. if 
\begin{equation}
 O_{\mathbbm{P}_0}(R_2)  \lesssim  n^{-\frac{2p}{2p+d}}
\end{equation}
\end{corollary}
\begin{proof}
This follows directly from the assumptions and Corollary \ref{errorbound}. 
\end{proof}

For the special case of coarsening at random problems, i.e. problems possessing doubly robust properties, the second-order remainder has a special form, which we briefly illustrated in appendix \ref{CATEImpl}. The property that only \textit{either} the coarsening mechanism \textit{or} the outcome regressions have to be consistently estimated has been extensively studied and exploited in the context of low-dimensional inference (see e.g. \cite{rubin2008empirical}).  In particular, it is well-known that this structure can be exploited to achieve root-n convergence of the average treatment effect when we can give a correctly specified parametric model for either component of the problem. Corollary \ref{ratethm} suggests that we can leverage similar properties also for high-dimensional inference (see discussion below) and that the  requirements on how much structure needs to be known can be substantially relaxed when the aim is only to attain oracle rates. 

\begin{remark}[Standard inference based on central limit theorems for EIF-based pseudo-outcome regression]
If we wanted to achieve convergence to a CLT, we would need that both the remainder term $R_2$ \textit{and} the oracle regression converge at parametric rates and are estimated using a suitable estimator. Clearly, by Corollary \ref{ratethm} this is impossible without making further assumptions on the structure of the problem and the nature of the regression estimator. However, for some data-adaptive nonparametric machine learning methods based on neighborhood  smoothing, there already  exist conditions that are milder than parametric assumptions which lead to a CLT asymptotically. As an example, we propose using the random forests described in \cite{wager2018estimation} and \cite{Athey2019}, which have some special properties (e.g. more randomness) that enable inference, for both regression stages.  In some cases, it may even be possible to outperform the estimation strategies discussed in \cite{Athey2019} by using the pseudo-outcomes we consider:  in section \ref{sectionrct} we use \cite{Athey2019}'s causal forests to show this empirically, and provide a more in-depth comparison of the estimators in appendix \ref{wadiscussion}.
\end{remark}

\subsubsection{Practical implications}
\paragraph{Plug-in estimation versus oracle estimation} It depends on the nature of the target function $\psi(x)$ whether the oracle rate in Corollary 2 is faster than the rate of the plug-in estimator.  Asymptotically,  pseudo-outcome regression has a major advantage over plug-in estimators particularly when we consider target parameters that are \textit{contrasts} of multiple parameters, e.g. differences such as CATE or nonlinear functions such as risk ratios. Plug-in estimators cannot adapt to the smoothness of the target function because they estimate components separately, while learners based on pseudo-outcomes estimate $\psi(x)$ directly, and can hence adapt. In the case of treatment effects, for example, it is often assumed that CATE is a much simpler function than the outcome regressions \citep{kunzel2019metalearners}, leading to a strong advantage of pseudo-outcome regression over plug-in estimators (see also the example in Appendix \ref{CATEImpl}). 

\paragraph{Attaining oracle rates in coarsening at random problems}
If coarsening at random problems are not fully nonparametric because we are able to exploit some knowledge of the nuisance functions, the proposed pseudo-outcome regression approach can easily attain oracle rates.  If, for example, we have information about the underlying (i) sparsity (or the subset of covariates that determine one of the nuisance functions), (ii) its parametric form or (iii) the exact nuisance function (e.g propensity scores), this can lead to the achievement of oracle rates because the nuisance estimation problem decreases in difficulty relative to the target estimation problem. This shows that incorporating domain knowledge can substantially help in achieving oracle rates. The most trivial case in which this is true is when the coarsening mechanism is known, e.g. if the propensity score is \textit{known} because we are in a RCT setting. In such settings we automatically achieve the oracle bound as $R_2$ is zero by construction.  This gives a new argument why simple plug-in estimation (without bias correction) is not asymptotically optimal \textit{especially} when coarsening mechanisms are known. Nonetheless,  domain knowledge does not have to be reflected in exact parametric forms. Even imposing some exclusion restrictions, such as excluding certain functional forms, can help reducing the difficulty of the problem by moving from a fully nonparametric to a semiparametric problem. Particularly in applications in biostatistics and medicine, such knowledge could be inferred from mechanistic understanding of diseases and domain knowledge on daily practice, e.g. by identifying variables that are more or less likely to influence missingness or selection bias. In applications from econometrics, such domain knowledge could stem from, for example, well-established theoretical models of micro- and macroeconomic theory.

With the discussion in this section, we have highlighted the immense potential of using EIFs to construct \textit{and} analyse estimators of target functions, in a way that is analogous to the well-studied low-dimensional case. In the following, we discuss the resulting learning algorithms in more detail.

\section{Learning algorithms}\label{learningalgs}
Based on the theoretical analysis presented in the previous section, we propose two general learning algorithms that can be used for estimation and inference for the broad class of problems of interest discussed in this paper.  We begin with our main algorithm, the general `IF-learner' which follows naturally from the pseudo-outcome regression set-up in the previous section. After that, we characterise a grouped version, the Group-IF-learner, which we believe could be of independent practical interest low sample size settings where we cannot rely on asymptotics but may have information on the coarsening mechanism.

\subsection{The IF-learner}
Our main proposal, the `IF-learner', is the algorithm that arises naturally from our assumptions  \ref{asssetup} on the pseudo-outcome regression set-up, and is presented in Algorithm \ref{alg1}. While it is motivated from the perspective of high-dimensional plug-in bias correction,  it can also be seen as a generalization of \cite{kennedy2020optimal}'s CATE estimator to a much broader class of target parameters. Here, we rely on cross-fitting (as discussed in e.g.  \cite{chernozhukov2018double} and \cite{kennedy2020optimal}) instead of sample splitting to be more efficient in our use of data.

\begin{algorithm}\caption{IF-learner}\label{alg1}
\begin{algorithmic}[1]
\State	Inputs: A sample $\mathcal{D}=\{O_i\}^n_{i=1}$, a target parameter $\psi$ with associated IF-based pseudo-outcome $D_{\psi, \mathbbm{P}}(O)$ which depends only on a subset of all nuisance parameters $Q \equiv Q(\mathbbm{P})$, a learning algorithm $\mathcal{A}$, and a number $K$ of cross-fitting folds to create
\State \textbf{First stage: plug-in model estimation}
\State split the sample $\mathcal{D}$ in $k$ non-overlapping folds
\For{$k \leftarrow 1:K$}
\State Fit nuisance models $\hat{Q}_{-k}=\mathcal{A}(\mathcal{D}_{-k})$ on all but the $k^{th}$ fold
\State Predict $\hat{D}_i = D_{\psi, \hat{\mathbbm{P}}}(O_i)$ for $O_i \in \mathcal{D}_k$ using the nuisance model $\hat{Q}_{-k}$
\EndFor 
\State \textbf{Second stage: plug-in bias correction step}
\State estimate $\psi(x)$ as a function of $x$ by regressing $\{\hat{D}_i\}^n_{i=1}$ on $\{X_i\}^n_{i=1}$ as $\hat{\psi}(x) = \mathcal{A}(\{\hat{D}_i, X_i\}^n_{i=1})$
\State Output: $\hat{\psi}$, a model that can output pointwise estimates of $\psi$
\end{algorithmic} 
\end{algorithm}

In the CATE case, the IF-learner reduces to \cite{kennedy2020optimal}'s `DR-learner' (and is similar to the two-stage estimators discussed in \cite{lee2017doubly} and \cite{fan2019estimation}). When learning CATE from observational data, the input learning algorithms for nuisance estimation would be two generic regression estimators for the potential outcome regressions, and an estimator for the propensity score.  If there is some additional knowledge on the structure of the problem, e.g. sparsity, covariates used for selection or functional forms, this could be incorporated by choosing only estimators that reflect this knowledge. As we discussed above, this can immediately lead to the algorithm achieving oracle rates. 

The algorithm can, however, be used for many more problem settings, particularly for other problems with coarsening at random structure. As we will show in section \ref{sectionrct}, this algorithm could, for example, be used to estimate personalised causal parameters that have not received as much attention in the machine learning literature, such as risk ratios. Further, we note that while the algorithmic description uses the same learning algorithm $\mathcal{A}$ for all components, that is certainly not necessary in practice, and separate learners could be used for each part of the problem. As discussed in section \ref{pracimp}, pointwise inference relying on a CLT using the IF-learner is possible only under more restrictive assumptions on the data generating process and for a restrictive choice of learning algorithms $\mathcal{A}$. The first stage of this algorithm also lends itself to the easy extension of incorporating super-learning \citep{polley2011super} for nuisance estimator model selection to improve finite sample performance, which would be an interesting idea to develop further.

\paragraph{Related approaches} Albeit motivated from the perspective of plug-in debiasing using influence functions, we arrive at an algorithm that is inherently related to the problem-generic algorithms described in \cite{semenova2017estimation} and \cite{foster2019orthogonal}, which exploit Neyman-Orthogonality \citep{chernozhukov2018double}. This is not surprising, given that both influence functions and Neyman-Orthogonality rely in their construction on orthogonal scores. Both papers have a different focus than the present paper: \cite{semenova2017estimation} develop semiparametric theory for estimation of generic target functions that are dependent on only a subset of a possibly high-dimensional vector of covariates using linear projections based on least squares series estimators. \cite{foster2019orthogonal} focus on bounding excess risk of loss-based learning algorithms with nuisance component, where the loss function satisfies Neyman-Orthogonality.

\subsection{The Group-IF-learner}
The second algorithm we propose here is motivated by the observation of \cite{Chernozhukov2018} that if we focus on estimating \textit{key features} of a function instead of the function itself, this can facilitate inference. Of the `key features' proposed in \cite{Chernozhukov2018} for treatment effect estimation, we are particularly interested in `Sorted Group Average Treatment Effects' (GATES), which provide a coarse summary of treatment effect heterogeneity by heterogeneity groups. In particular, the authors propose the following algorithm for constructing GATES based on $G$ groups and known propensity scores:
\begin{enumerate}
\item Split the data in two samples, an auxiliary sample $\mathcal{D}_A$ and an estimation sample $\mathcal{D}_E$.
\item First stage: Fit a treatment effect model $\tilde{\psi}(x)$ and a baseline model $\tilde{\mu}(x)=\mathbbm{E}[Y|X=x, W=0]$ on $\mathcal{D}_A$
\item Second stage: Predict $\tilde{\psi}(X_i)$ for the estimation sample $\mathcal{D}_E$, and build groups of individuals that are most similar in terms of their treatment effects by grouping observations in $G$ groups according to the empirical quantiles of their estimated treatment effects $\{\tilde{\psi}(X_i)\}_{i\in \mathcal{D}_E}$. Within these groups, estimate treatment effects by orthogonalised weighted regression, possibly including $\tilde{\mu}(x)$ as a covariate to improve precision. 
\end{enumerate}

Because of sample splitting, often also referred to as \textit{honesty} \citep{Athey2016}, and known propensity scores, the within-group estimates are unbiased and standard inference is possible. However, we see two points of possible improvements to this algorithm.  First, if the first stage model $\tilde{\psi}(x)$ suffers from high degrees of plug-in bias, then the groups are created based on noise only -- which is irreversible because partitioning is a hard-thresholding operation. While this does not invalidate inference, it could substantially decrease the usefulness of the proposed algorithm. Second, the regression adjustment procedure for the second stage proposed in \cite{Chernozhukov2018} is not optimal from the standpoint of efficiency, as we have shown the true form of the efficient estimator for group-averaged target parameters in Proposition \ref{partprop}, which is a within-group AIPW estimator in the CATE case.

The `Group-IF-learner' we propose is an adapted version of the GATES approach proposed in \cite{Chernozhukov2018}, in which we use our IF-learner to remove the first-order plug-in bias in the first stage of the algorithm to obtain a better approximation $\tilde{\psi}(x)$ for grouping. We then adapt the second stage by using the efficient estimators based on the efficient influence function. Note that, because we are considering groups in the second stage (as in Proposition \ref{partprop}), the efficient influence function exists in the standard sense. This leads to the algorithm described in Algorithm \ref{alg2}.

\begin{algorithm}\caption{Group-IF-learner}\label{alg2}
\begin{algorithmic}[1]
\State	Inputs: All inputs of Algorithm \ref{alg1}. $G$, the number of groups to be created
\State Split the sample $\mathcal{D}$ into two non-overlapping groups $\mathcal{D}_A$ and $\mathcal{D}_E$
\State \textbf{First stage: Learning step}
\State Fit nuisance models $\hat{Q}_{A}=\mathcal{A}(\mathcal{D}_{A})$
\State Create auxiliary, plug-in debiased model $\tilde{\psi}(x)$ by using the IF-learner (Algorithm \ref{alg1}) on $\mathcal{D}_A$
\State \textbf{Second stage: Estimation step}
\State Split $\mathcal{D}_E$ into $G$ groups by using the empirical quantiles of $\tilde{\psi}(x)$
\State Estimate $\Psi_g$ for each group $g$ by using the empirical average $\hat{\Psi}_k = \frac{1}{n_g}\sum^n_{i \in g}D_{\psi, \hat{\mathbbm{P}}}(O_i)$ and an unbiased empirical variance estimate $\frac{1}{n_g(n_g-1)}\sum_{i \in g} (D_{\psi, \hat{\mathbbm{P}}}(O_i)- \hat{\Psi}_g)^2 $
\State Output: $\{\hat{\Psi}_g\}^G_{g=1}$ and associated variance estimates. 
\end{algorithmic} 
\end{algorithm}

Due to sample-splitting, the within-group estimates are asymptotically unbiased and Gaussian if the nuisance parameters in the first stage are estimated at parametric rates. Thus, this discretisation is not useful in high-dimensional cases where there is no information available. However, in randomised experiments (as was the setting in \citet{Chernozhukov2018}) and other problems where the coarsening mechanism is \textit{known}, this estimator is unbiased even in finite samples due to the double robustness property of the EIF and asymptotically Gaussian \citep{rubin2008empirical}, and thus allows for standard inference. Thus, our second algorithm is suited more to high-information, low-sample size settings -- which is the case for RCTs. 

\begin{remark}[Accounting for uncertainty induced by sample splitting]
To obtain final GATES estimates, \cite{Chernozhukov2018} repeat the grouping procedure for multiple splits of the data into estimation and auxiliary samples, resulting in a number of point estimates and confidence intervals that are aggregated into one estimate by group and associated confidence intervals by a tailored and new assumption-free procedure developed in \cite{Chernozhukov2018}. We did not yet attempt to incorporate this idea further, however because our learner should be more stable due to first-stage bias correction and second stage efficient estimation, we conjecture that this would only improve our algorithm relative to that of \cite{Chernozhukov2018}.\end{remark}

\section{Simulation study: Estimating treatment effects}\label{sectionrct}
In this section, we return to our main motivating example -- estimating (possibly) heterogeneous treatment effects from experimental and observational data -- with two simple simulation studies.

\subsection{Assumptions and set-up}
 Recall the problem setting of binary treatment effect estimation: we have observations $O_i=(Y_i, X_i, W_i)$ with binary treatment $W_i \in \{0,1\}$ assigned according to propensity score $\pi(x)=\mathbbm{P}_0(W=1|X=x)$. We are interested in  estimating quantities such as the CATE, $\tau(x)=\mathbbm{E}_{F_0}[Y(1) - Y(0)|X=x]$, which are functions of the two potential outcome regressions.  The CATE and its population average, the average treatment effect (ATE), $\Upsilon = \mathbbm{E}_{\mathbbm{P}_0}[\tau(x)]$ and other such causal parameters are identifiable from data under three assumptions:

\begin{assumption}\textbf{Consistency: } If an individual is assigned treatment $W=w$, we observe the associated potential outcome $Y=Y(w)$
\end{assumption}

\begin{assumption}
\textbf{Unconfoundedness: } Treatment is randomised, with treatment assignment probability $\pi(x)$ depending only on the covariates. That is, 
\begin{equation}
Y(0), Y(1) \indep W | X 
\end{equation}
\end{assumption}

\begin{assumption}
\textbf{Overlap:} 
The treatment assignment probability is uniformly bounded away from zero and one, that is, 
\begin{equation}
0 < \pi(x) < 1 \text{ for all } x \in \mathcal{X} 
\end{equation}
i.e. we observe each potential outcome with positive probability.
\end{assumption}

These identifying assumptions are in general untestable, and whether they hold in practice should be determined by domain experts. Further, we also need to assume that both $Y(0)$ and $Y(1)$ have finite variance under $\mathbbm{P}_0$. Last, we note that within our framework we also make the assumption that our observations $\{O_i\}^{n}_{i=1}$ are i.i.d. -- that is, they are mutually independent\footnote{ In clinical trials, this might not always be the case, as treatment assignment is often randomised under the restriction that a fixed number of participants is assigned to each trial arm in order to avoid severe imbalances \citep{friedman2015fundamentals}, making the assignments dependent \citep{rubin2008covariate}. Such dependent sampling scheme in the context of empirical efficiency maximisation is discussed in \cite{rubin2008covariate}, and it would be interesting to investigate whether their results could be used for heterogeneous inference within our framework. For now, we restrict our attention to the more common i.i.d. assumption.}. 

Recall from example 3.2. that the uncentered EIF of the ATE (and hence our pseudo-outcome) is given by
\begin{equation*}
D_{\Psi, \mathbbm{P}}(O)=  \left(\frac{W}{\pi(X)}- \frac{(1-W)}{1-\pi(X)}\right) Y + \left[\left(1 - \frac{W}{\pi(X)}\right) \mu_1(x)-\left(1 - \frac{1-W}{1-\pi(X)}\right)\mu_0(X)\right]
\end{equation*}
with $\mu_w(x) = \mathbbm{E}_{\mathbbm{P}}[Y(w)|X=x)] = \mathbbm{E}_{\mathbbm{P}}[Y|W=w, X=x]$. Note that when the propensity score is \textit{known}, the estimator for the ATE based on this EIF is unbiased regardless of the quality of the two regression function estimators -- because the first term is the standard inverse-propensity weighted Horvitz-Thompson estimator \citep{horvitz1952generalization} which is unbiased if $\pi(x)$ is known, and the second term has expectation zero \citep{rubin2008empirical}. If $\pi(x)$ is \textit{not} known then this estimator is doubly-robust, as it will give consistent treatment effects estimates if \textit{either} $\pi(x)$ or $\mu_w(x)\text{, } w \in \{0, 1\}$ can be consistently estimated.
 
\subsubsection{Targets that are nonlinear functions of the potential outcome regressions}
In clinical trials and other experimental studies, in particular when the outcome variable is binary, the target parameter of interest is often not the (C)ATE $\mu_1(x)-\mu_0(x)$, but a transformation $f(\mu_0(x), \mu_1(x))$ of the expected potential outcomes such as a risk ratio \citep{rubin2008empirical}. Due to the generality of our approach, our framework can naturally handle these cases as well. \cite{rubin2008empirical} show that the EIF for generic substitution-based transformation $f(\mu_0, \mu_1)$ is given by
\begin{equation*}
IF_{f(\mu_0, \mu_1), \mathbbm{P}}(O) = IF_{\mu_0, \mathbbm{P}}(O) \frac{\partial}{\partial \mu_0} f(\mu_0, \mu_1)+ IF_{\mu_1, \mathbbm{P}}(O)\frac{\partial}{\partial \mu_1} f(\mu_0, \mu_1)
\end{equation*}
where $IF_{\mu_w, \mathbbm{P}}$ is the \textit{centered} efficient influence function of the potential outcome mean $\mu_w$. This idea can  be naturally incorporated as an alternative pseudo-outcome in our learning algorithms. This gives our framework a considerable advantage over most learning algorithms for heterogeneous treatment effect estimation, which typically output only conditional average treatment effects. For risk ratios $f_{RR}(\mu_0(x), \mu_1(x))=\frac{\mu_1(x)}{\mu_0(x)}$  and odds ratios $f_{OR}(\mu_0(x), \mu_1(x))= \frac{\mu_1(x)}{1-\mu_1(x)}/ \frac{\mu_0(x)}{1-\mu_0(x)}$, and the relevant partial derivatives \\ $[\frac{\partial}{\partial \mu_1}  f(\mu_0, \mu_0), \frac{\partial}{\partial \mu_0} f(\mu_0, \mu_0)]$ are given by $[\mu_0^{-1}, -\mu_1\mu_0^{-2}]$ and $[\frac{1-\mu_0}{(1-\mu_1)^2\mu_0}, \frac{-\mu_1}{(1-\mu_1)\mu_0^2}]$, respectively \citep{rubin2008empirical}. 

In the experiments, we provide an example using risk-ratios (RR), for which we use the following unbiased EIF-inspired pseudo-outcome:
\begin{equation*}
\begin{split}
D_{RR, \mathbbm{P}}(O)= \frac{1}{\mu_0(X)}\left[\frac{W}{\pi(X)}Y + \left(1 - \frac{W}{\pi(X)}\right)\mu_1(X) - \mu_1(X) \right] \\- \frac{\mu_1(X)}{\mu^2_0(X)} \left[\frac{1-W}{1-\pi(X)}Y + \left(1 - \frac{1-W}{1-\pi(X)}\right)\mu_0(X) - \mu_0(X) \right] + \frac{\mu_1(X)}{\mu_0(X)}
\end{split}
\end{equation*}
 
\subsection{Simulation study 1: Learning with known propensity scores}
In this section we illustrate the performance differences between plug-in estimators and bias-corrected estimators in experimental studies when propensity scores are \textit{known}. To do so, we revisit the one-dimensional toy-example used in \cite{kennedy2020optimal}, which is based on a difficult piecewise polynomial baseline effect function $\mu_0(x)$ from \cite{gyorfi2006distribution}, while the treatment effect $\tau(x)=\tau=0$ is not only constant but also zero. This set-up illustrates very well how plug-in bias affects data-adaptive target estimates even when the covariates are only one-dimensional. Thus, this simple example serves the purpose of illustrating the impact of plug-in debiasing.

\cite{kennedy2020optimal} used this example to highlight the difficulty of plug-in learners without bias correction step to handle observational problems in which there is (i) a very difficult baseline model and (ii) strong (unknown) selection bias, for which a propensity score has to be estimated. Since in RCTs and other experimental studies the randomisation probabilities are \textit{known}, plug-in debiased methods are very valuable because -- as we have tried to highlight throughout section \ref{sectheor} --  when one of the nuisance functions is known, we asymptotically achieve oracle rates and hence optimal convergence without further assumptions. Below, we illustrate this idea using some evidence from simulations. 

\subsubsection{Large samples and the IF-learner}\label{exp1}
For the example of CATE estimation from experimental data, we present results using smoothing splines in Figure \ref{experiment1}. Refer to appendix \ref{expsetup} for a detailed description of experimental settings.  The three settings allow us to to show that (i) our learner converges faster than an uncorrected plug-in estimator even under pure randomisation ($\pi(x)=\pi=0.5$), (ii) performs a lot better than plug-in estimation when propensity scores are known even for relatively modest sample sizes and (iii) for fixed training sample size (n=500) even performs better when randomisation is assumed, but there is some form of unknown selection bias. The latter might sound counter-intuitive but hinges on the observation that the second-stage regression also acts like a regularizer on the first stage regression output.

\begin{figure}[!htb]
 \makebox[\textwidth][c]{\includegraphics[width=1.25\textwidth]{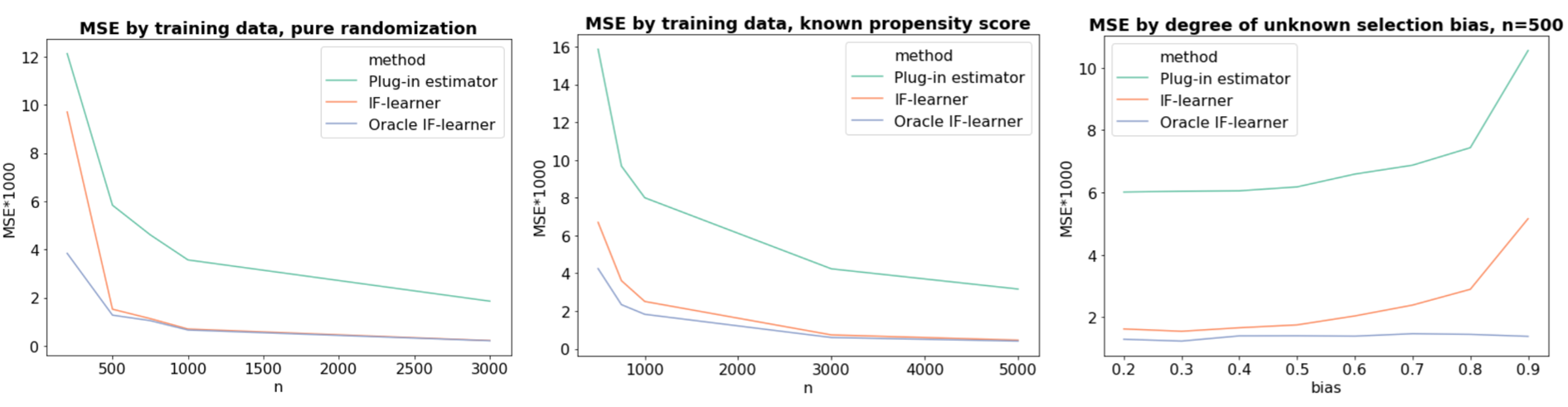}}
\caption{Simulation performance of plug-in model, plug-in bias corrected model (IF-learner) and its oracle version for CATE estimation under different settings. MSE averaged over 500 independent runs.}\label{experiment1}
\end{figure}

To illustrate that our framework can handle other causal parameters, we use the same set-up (difficult baseline, no treatment effect) but this time to estimate risk-ratios as discussed in the previous section. To do so, we use the same baseline function, but now as the success probability of a Bernoulli-GLM. Refer to appendix \ref{expsetup} also for the exact simulation specification and estimation strategies used in this scenario. 

The results in Figure \ref{experiment2} highlight that in the RR setting a trade-off between small and large sample performance becomes much more obvious than in the CATE setting above: because the outcome $Y$ is binary while the pseudo-outcome $D_{RR, \mathbbm{P}}(O)$ in the second stage is continuous, the estimation problem is more difficult and substantially more data is needed to remove error induced in the second stage regression due to small sample variance. 

\begin{figure}[!htb]
\centering
\includegraphics[width=\textwidth]{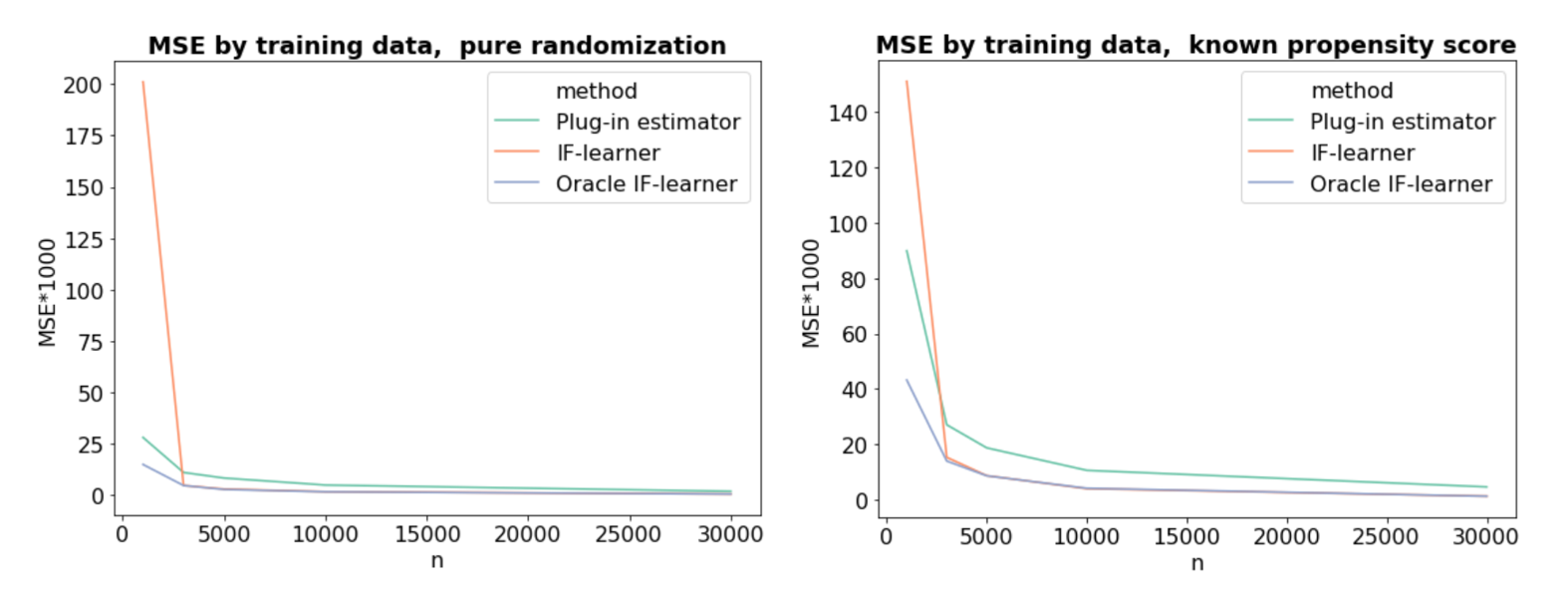}
\caption{Simulation performance of plug-in model, plug-in bias corrected model (IF-learner) and its oracle version for RR estimation under different settings. MSE averaged over 500 independent runs.}\label{experiment2}
\end{figure}

\subsubsection{Small samples and the Group-IF-learner}\label{exp2}
We now turn to illustrating the performance of our Group-IF-learner in smaller samples, which is often the setting for RCTs. We use the original CATE example with strong selection on observables according to a known propensity score for illustrative purposes. Simulation results are presented in Table \ref{res1}. We compare the base learner of \cite{Chernozhukov2018}, where we only consider the version \textit{without} repeated sample splitting (Row 1) to our Group-IF-learner (Row 5) as well as three versions with different combinations of first- and second-stage estimators, to analyse the performance differences in more detail.

We make three interesting observations. First, by comparing the first row to all other rows of Table \ref{res1}, we observe that the base learner of \cite{Chernozhukov2018} is outperformed by one version of our estimator for all training sample sizes $n<1000$. Second, we observe that the best estimator depends on the sample size: for very small samples, $n\leq 100$, estimators using the Horvitz-Thompson estimator within each group work best, while the EIF estimators perform substantially better from $n=500$ onwards. This indicates that there is a clear trade-off between removing plug-in bias of the second stage plug-in estimator, and inducing larger variance by doing so in small samples. Finally, by comparing rows 3 and 5, we observe that removing the plug-in bias of the first stage estimator is not the main source of gain -- rather, for moderate sample size, the gain relative to \cite{Chernozhukov2018} seems to mainly be driven by replacing the second stage orthogonalised weighted regression  with the group-wise efficient estimator derived in Proposition \ref{partprop}.

\begin{table}[!htb]
\begin{threeparttable}
\small
\setlength\tabcolsep{4.5pt}
\begin{tabular}{l|llllll}
\toprule
\midrule
Method/ Training observations &  100                                                     & 200                                                     & 500                                                     &750 & 1000                                                     & 2000                                                     \\
\toprule
\midrule
\cite{Chernozhukov2018}        & \begin{tabular}[c]{@{}l@{}}1.329\\ (0.182)\end{tabular} & \begin{tabular}[c]{@{}l@{}}0.562\\ (0.162)\end{tabular} & \begin{tabular}[c]{@{}l@{}}0.021\\ (0.002)\end{tabular} & \begin{tabular}[c]{@{}l@{}}0.010\\ (0.001)\end{tabular}      & \begin{tabular}[c]{@{}l@{}}0.006\\ ($<0.001$)\end{tabular} & \begin{tabular}[c]{@{}l@{}}0.003\\ ($<0.001$)\end{tabular} \\
\begin{tabular}[c]{@{}l@{}}Group-IF-learner: \\Plug-in 1st stage, HT 2nd \end{tabular}     & \begin{tabular}[c]{@{}l@{}}0.844\\ (0.041)\end{tabular} & \begin{tabular}[c]{@{}l@{}}0.420\\ (0.016)\end{tabular} & \begin{tabular}[c]{@{}l@{}}0.169\\ (0.005)\end{tabular} & \begin{tabular}[c]{@{}l@{}}0.117\\ (0.003)\end{tabular}      & \begin{tabular}[c]{@{}l@{}}0.088\\ (0.003)\end{tabular}      & \begin{tabular}[c]{@{}l@{}}0.043\\ (0.001)\end{tabular}      \\
\begin{tabular}[c]{@{}l@{}}Group-IF-learner: \\ Plug-in 1st stage, EIF 2nd \end{tabular}   & \begin{tabular}[c]{@{}l@{}}1.250\\ (0.353)\end{tabular} & \begin{tabular}[c]{@{}l@{}}1.506\\ (1.010)\end{tabular} & \begin{tabular}[c]{@{}l@{}}0.014\\ (0.001)\end{tabular} & \begin{tabular}[c]{@{}l@{}}0.009\\ ($<0.001$)\end{tabular} & \begin{tabular}[c]{@{}l@{}}0.007\\ ($<0.001$)\end{tabular} & \begin{tabular}[c]{@{}l@{}}0.003\\ ($<0.001$)\end{tabular} \\
\begin{tabular}[c]{@{}l@{}}Group-IF-learner: \\IF-learner 1st stage, HT 2nd \\\end{tabular} & \begin{tabular}[c]{@{}l@{}}0.890\\ (0.040)\end{tabular} & \begin{tabular}[c]{@{}l@{}}0.401\\ (0.015)\end{tabular} & \begin{tabular}[c]{@{}l@{}}0.162\\ (0.005)\end{tabular} & \begin{tabular}[c]{@{}l@{}}0.117\\ (0.003)\end{tabular}      & \begin{tabular}[c]{@{}l@{}}0.086\\ (0.003)\end{tabular}      & \begin{tabular}[c]{@{}l@{}}0.044\\ (0.001)\end{tabular}      \\
\begin{tabular}[c]{@{}l@{}}Group-IF-learner: \\IF-learner 1st stage, EIF 2nd\end{tabular} & \begin{tabular}[c]{@{}l@{}}2.184\\ (0.883)\end{tabular} & \begin{tabular}[c]{@{}l@{}}0.338\\ (0.086)\end{tabular} & \begin{tabular}[c]{@{}l@{}}0.015\\ (0.001)\end{tabular} & \begin{tabular}[c]{@{}l@{}}0.009\\ ($<0.001$)\end{tabular} & \begin{tabular}[c]{@{}l@{}}0.006\\ ($<0.001$)\end{tabular} & \begin{tabular}[c]{@{}l@{}}0.003\\ ($<0.001$)\end{tabular}\\
\midrule
\bottomrule
\end{tabular}
\begin{tablenotes}[para,flushleft]
\begin{footnotesize}
First row corresponds to the GATES estimator proposed in \cite{Chernozhukov2018}, last row corresponds to our Group-IF-learner. Rows 2 to 5 change either the first or the second stage of algorithm \ref{alg2}. 
HT denotes Horvitz-Thompson estimator, EIF denotes efficient influence function estimator (here: AIPW estimator). The out-of-sample MSE for 1000 randomly generated test observations is averaged over 500 simulations, standard error of the mean in parentheses.
\end{footnotesize}
\end{tablenotes}
\end{threeparttable}
\caption{MSE of different methods for group-based treatment effect inference for different number of training observations}\label{res1}
\end{table}

\subsection{Simulation study 2: Learning from higher-dimensional observational data}
In this section, we revisit the toy-examples used in \cite{Athey2019} to show how we can use our ideas to learn from higher-dimensional observational data with unknown selection bias using machine learning methods.  Here, we use \cite{Athey2019}'s random forests in both stages of our algorithm. We compare the IF-learner to an uncorrected plug-in estimator as in the previous section and to \cite{Athey2019}'s causal forests (CFs) (for a discussion on how exactly our algorithm differs from CFs, refer to appendix \ref{wadiscussion}). To allow for fair comparison with \cite{Athey2019}, we use out-of-bag predictions instead of cross-fitting for the IF-learner. 

In table \ref{atheyres}, we consider multiple experimental settings based on the toy-examples in \cite{Athey2019}. For each, we let $X \in [0,1]^{10}$, and use different combinations of confounding and treatment effect specifications. The settings in rows (1)-(3) were investigated directly in \cite{Athey2019}, rows (4)-(5) are new as they contain treatment effects that are not simply additive. For a full description of the data generating processes, refer to appendix \ref{expsetup2}. 

\begin{table}[!htb]
\centering
\begin{threeparttable}
\small
\setlength\tabcolsep{4pt}
\begin{tabular}{lll|lll}
\toprule
\midrule
Confounding                                                                                      & Effect specification                                                                                                    & \begin{tabular}[c]{@{}l@{}}n\\ (Training)\end{tabular}     & Plug-in                                                 & \begin{tabular}[c]{@{}l@{}}Causal\\ Forest\end{tabular}                                                       & IF-learner                                                         \\
\toprule
\midrule
\begin{tabular}[c]{@{}l@{}}$\mu_0(x)=0$\\ $\pi(x)=0.5$\end{tabular}                              & $\tau(x)=\xi(x_1)\xi(x_2)$& 800  & \begin{tabular}[c]{@{}l@{}}0.063\\ (0.001)\end{tabular} & \begin{tabular}[c]{@{}l@{}}0.100\\ (0.003)\end{tabular}            & \begin{tabular}[c]{@{}l@{}}0.095\\ (0.002)\end{tabular}            \\
                                                                                               &                                                                                                                        & 1600 & \begin{tabular}[c]{@{}l@{}}0.037\\ (0.001)\end{tabular} & \begin{tabular}[c]{@{}l@{}}0.054\\ (0.001)\end{tabular}            & \begin{tabular}[c]{@{}l@{}}0.056\\ (0.001)\end{tabular}            \\ \midrule
\begin{tabular}[c]{@{}l@{}}$\mu_0(x)=2x_3 -1$\\ $\pi(x)= 0.25(\beta_{2,4}(x_3)+1)$\end{tabular} & $\tau(x)=0$                                                                                                             & 800  & \begin{tabular}[c]{@{}l@{}}0.040\\ (0.001)\end{tabular} & \begin{tabular}[c]{@{}l@{}}0.011\\ (0.001)\end{tabular}            & \begin{tabular}[c]{@{}l@{}}0.020\\ (0.001)\end{tabular}            \\
                                                                                              &                                                                                                                       & 1600 & \begin{tabular}[c]{@{}l@{}}0.027\\ (0.001)\end{tabular} & \begin{tabular}[c]{@{}l@{}}0.007\\ (\textless{}0.001)\end{tabular} & \begin{tabular}[c]{@{}l@{}}0.014\\ (\textless{}0.001)\end{tabular} \\ \midrule
\begin{tabular}[c]{@{}l@{}}$\mu_0(x)=2x_3 -1$\\ $\pi(x)= 0.25(\beta_{2,4}(x_3)+1)$\end{tabular}                                                & $\tau(x)=\xi(x_1)\xi(x_2)$ & 800  & \begin{tabular}[c]{@{}l@{}}0.149\\ (0.003)\end{tabular} & \begin{tabular}[c]{@{}l@{}}0.110\\ (0.003)\end{tabular}            & \begin{tabular}[c]{@{}l@{}}0.107\\ (0.003)\end{tabular}            \\
                                                                                               &                                                                                                                     & 1600 & \begin{tabular}[c]{@{}l@{}}0.091\\  (0.002)\end{tabular} & \begin{tabular}[c]{@{}l@{}}0.057\\ (0.001)\end{tabular}            & \begin{tabular}[c]{@{}l@{}}0.060\\ (0.001)\end{tabular}            \\ \midrule
                                              \begin{tabular}[c]{@{}l@{}}$\mu_0(x)=2x_3 -1$\\ $\pi(x)= 0.25(\beta_{2,4}(x_3)+1)$\end{tabular}                                                     &     $\tau(x)=3\mu_0(x)$                                                                                                                    & 800  & \begin{tabular}[c]{@{}l@{}}0.085\\ (0.002)\end{tabular} & \begin{tabular}[c]{@{}l@{}}0.149\\ (0.004)\end{tabular}            & \begin{tabular}[c]{@{}l@{}}0.126\\ (0.003)\end{tabular}            \\
                                                                                                 &                                                                                                                         & 1600 & \begin{tabular}[c]{@{}l@{}}0.050\\ (0.001)\end{tabular} & \begin{tabular}[c]{@{}l@{}}0.074\\ (0.002)\end{tabular}            & \begin{tabular}[c]{@{}l@{}}0.070\\ (0.001)\end{tabular}            \\ \midrule
                                                             \begin{tabular}[c]{@{}l@{}}$\mu_0(x)=2x_3 -1$\\ $\pi(x)= 0.25(\beta_{2,4}(x_3)+1)$\end{tabular}                                           &                                                                     $\tau(x)=\mu_0(x)\xi(x_1)\xi(x_2)$                                                      & 800  & \begin{tabular}[c]{@{}l@{}}0.306\\ (0.003)\end{tabular} & \begin{tabular}[c]{@{}l@{}}0.401\\ (0.005)\end{tabular}            & \begin{tabular}[c]{@{}l@{}}0.335\\ (0.004)\end{tabular}            \\
                                                                                                 &                                                                                                                         & 1600 & \begin{tabular}[c]{@{}l@{}}0.200\\ (0.002)\end{tabular} & \begin{tabular}[c]{@{}l@{}}0.252\\ (0.003)\end{tabular}            & \begin{tabular}[c]{@{}l@{}}0.215\\ (0.003)\end{tabular}  \\ \midrule
\bottomrule
\end{tabular}
\begin{tablenotes}[para,flushleft]
\begin{footnotesize}
The out-of-sample MSE is averaged over 200 runs with 1000 randomly generated test observations. Standard errors in parentheses. All random forests are trained with 2000 trees and defaults as implemented in the package \url{grf.}
Here, $\xi(x_i)=1+\frac{1}{1+\exp(-20*(x_i-1/3))}$ and $\beta_{a,b}$ is the density of the beta-distribution with parameters $a,b$. 
\end{footnotesize}
\end{tablenotes}
\end{threeparttable}
\caption{MSE of different methods for CATE estimation for different number of training observations and data generating processes}\label{atheyres}
\end{table}

We make a number of interesting observations: First, none of the learners outperforms all others across all settings. The simple plug-in learner outperforms CFs and the IF-learner in some settings -- namely when the set-up is either very simple or very difficult. We conjecture that, similar to the previous section, this is mainly the case when asymptotics have not yet kicked in and is likely to change with increasing sample size. When comparing CF and IF-learner, we observe that CFs perform best when treatment effect and baseline are perfectly separable. The IF-learner performs relatively better in smaller samples, and when treatment effects are not simply additive. 

\subsection{Discussion}
In all experiments, when comparing simple plug-in learners and IF-learners, there is a clear trade-off in terms of finite/small sample variance and large sample performance. The results discussed in section \ref{pracimp} hold asymptotically, however, in small samples the variance induced by using two-stage estimation using imputed pseudo-outcomes can outweigh the asymptotic benefits of our plug-in debiasing procedure relative to simple plug-in estimators. Therefore, it would be of great interest to investigate whether approaches such as super-learning \citep{polley2011super} or variance reduction techniques could help in finding the right type and degree of regularisation for first stage models. Further, it would be interesting -- similar to the discussion of the limits of Bayesian nonparametric learning of CATE in moderate to small samples contained in \cite{alaa2018limits} -- to investigate theoretically when plug-in bias correction is expected to be most useful in finite samples.

\section{Conclusion}
In this paper, we proposed IF-learning, a framework for learning structural target functions based on influence functions using generic machine learning methods. Using a very simple and intuitive idea, based on replacing the uncentered efficient influence function with approximations when it does not exist, we considered efficient estimation, plug-in bias and the possibility to perform inference using machine learning estimators for entire target functions. We hope that our approach can help the field of applied statistics realise more of the inherent potential of the recent advances in machine learning. 

We also believe that there are many interesting and new research directions within this learning framework. In particular, we believe that much can be learned from further exploring  connections of our proposals, which are mainly based on the first-order efficient influence function of a target parameter, with other concepts in low- and high-dimensional inference that have been proposed in the last 20 years. In particular, we believe that it might be interesting to further incorporate ideas from the literature on targeted maximum likelihood estimation \citep{van2011targeted}, empirical efficiency maximisation \citep{rubin2008empirical}, super learning \citep{polley2011super}, Neyman-Orthogonality and high-dimensional locally robust semiparametric estimation (e.g. \cite{chernozhukov2018double, semenova2017estimation, chernozhukov2020locally}), and approaches to achieving fast rates using cross-fitting \citep{newey2018cross} and higher-order influence functions \citep{robins2008higher}. We hope to be able to explore some of these in future work.

\bibliography{dissertation}
\bibliographystyle{apalike}
\newpage
\appendix
\section*{Appendices}

\section{Additional literature review}\label{litreview}
We draw our inspiration from ideas proposed in the econometrics, biostatistics, semiparametric statistics, causal inference and (statistical) machine learning communities. Due to the sheer breadth of related topics, this review cannot be exhaustive. Instead, we focus on key ideas from these fields that shaped the approach presented in this paper. Throughout the paper, we highlighted similarities and differences with related literature in more detail. 

We start with the field of econometrics, where the last five years have seen exciting developments in terms of using machine learning for estimation of and inference on (heterogeneous) treatment effects -- which motivated much of this paper. The ideas most relevant to us were shaped by the focus on heterogeneous inference using tree-based methods in \citet{Athey2016}, \citet{wager2018estimation} and \citet{Athey2019}, on the one hand, and the focus on exploiting generic machine learning methods for the estimation of structural target parameters in the recent work of Chernozhukov and colleagues (e.g.  \cite{chernozhukov2018double}, \cite{semenova2017estimation}  and \cite{chernozhukov2020locally}) on the other. While the former builds on classical ideas from the Neyman-Rubin Potential outcomes framework (\cite{neyman1923applications}, \cite{rubin1978bayesian}),  developed to maturity for average treatment effect estimation over the last 30 years in econometrics \citep{athey2017estimating}, much of the latter builds on a property the authors refer to as \textit{Neyman orthogonality}, a notion relying on estimating equations induced by orthogonal scores \citep{chernozhukov2018double}. 

The latter also motivated the `orthogonal statistical learning' framework presented in \cite{foster2019orthogonal}, which, like our proposed framework, considers generic target functions, albeit with a different focus: Instead of enabling high-dimensional parameter estimation and inference (the goal of the present paper), \cite{foster2019orthogonal} focus on using orthogonality to optimally bound excess risk and construct loss functions in such learning problems. Further, many of the questions we initially tried to tackle in this paper were motivated by the inspiring discussion of the fundamental limits of using generic machine learning for heterogeneous treatment effect inference from experimental studies in \citet{Chernozhukov2018}. We also adapt a proposal of \cite{Chernozhukov2018} (which is itself built on findings of \cite{genovese2008adaptive})  -- to learn only `key features' of a function instead of the function itself to facilitate inference -- for one of our learning algorithms.

In biostatistics, there exist slightly older ideas contained in the more mature literature on targeted learning, which has been focused on developing smart strategies for plug-in estimation using parametric submodels and estimating equations based on influence functions. We took much of the excellent formalisation of problems of interest in statistics from this literature, but did not yet explore in detail further connections between our approach and the broad literature on targeted maximum likelihood estimators (TMLEs) as originally proposed in \citet{van2006targeted}. In particular, the TMLE approach presented in \citet{van2018cv} is provides an alternative solution to the infinite-dimensional plug-in bias problem we are trying to solve. How to combine these ideas would be very interesting to explore in future work. Additionally, while we also do not yet incorporate the idea of super learning \citep{polley2011super}, this idea is complementary to our approach. 

Instead of the approaches discussed above, we ultimately found the notion of (low-dimensional) plug-in estimation and bias correction as it characteristic for \citet{robins2017minimax}'s work on (higher-order) influence functions the conceptually simplest theoretical framework to work in. The simple mathematical intuition and elegance behind using efficient influence functions and  Taylor expansion/Newton-Raphson Step-like bias correction procedures naturally gave rise to the idea that we propose in this paper. We also build on the ground-breaking work of Robins and colleagues that lead to the characterisation of many coarsening at random problems in semiparametric statistics in terms of their efficient influence functions 30 years ago, as summarised in \cite{van2003unified} and \cite{tsiatis2007semiparametric}. 

Because of their inherent mathematical elegance, influence functions have recently begun to rise to popularity also in the machine learning community. In \cite{alaa2020discriminative} and \cite{alaa2020frequentist}, they were used to characterise uncertainty in deep learning, albeit using a different property of the influence function more natural to the field of robust statistics originally coined by \cite{hampel1980robust}. Using the notion of influence function that we are interested in here, \cite{alaa2019validating} use influence functions to adjust for plug-in bias when choosing between different methods for treatment effect estimation via cross-validation. The original questions posed in this paper were inspired partially by this work. 

Finally, our approach relates to some ideas from the causal inference community within machine learning, due to its focus on building model-agnostic algorithms for CATE estimation, often referred to as meta-learners (e.g. \cite{nie2017quasi} and \cite{kunzel2019metalearners}). The recently independently proposed doubly robust CATE meta-learner of \citet{kennedy2020optimal}, as well as similar two-step estimators presented previously in \cite{lee2017doubly} and \cite{fan2019estimation}, are special cases of one of the algorithms that we propose in this paper, as our `IF-learner' reduces to the same estimator in the context of CATE estimation. A similar idea to ours in the special case of estimating treatment effects of continuous treatments was also proposed in \cite{kennedy2017nonparametric}. Additionally, an important and more general result on error bounds for pseudo-outcome regression proposed in \cite{kennedy2020optimal} paved the way for most of the convergence rate discussions later in this paper. \cite{kennedy2020optimal} also discusses the fundamental limits of CATE estimation in fully nonparametric settings in terms of learning rates, and a similar discussion for Bayesian nonparametric settings was presented in \cite{alaa2018limits} and \cite{alaa2018bayesian}.

\section{Technical appendix}
\subsection{Proof of Proposition 1}\label{proofprop1}
\begin{proof}(By contradiction, using the efficient RAL estimator)\\ We begin by noting that, because our target parameters are (linear combinations of) conditional means 
\begin{equation*}
\Psi = \mathbbm{E}_{{\mathbbm{P}}_0}[\psi(x)] = \mathbbm{E}_{\mathbbm{P}_0}[\mathbbm{E}_{\mathbbm{P}_0}[\psi(x)|X \in A^\pi_k]] = \sum^{K^\pi}_{k=1} \mathbbm{P}_0(X \in A^\pi_k)\Psi^\pi_k
\end{equation*}
so the full population mean is simply a weighted average of the partition-means. Therefore, we can write 
\begin{equation*}
\dot{\Psi} = D_{\Psi, \mathbbm{P}_0}(O) - \Psi = \sum^{K^\pi}_{k=1} \bigg(\mathbbm{1}(X \in A^\pi_k)D_{\Psi, \mathbbm{P}_0}(O) - \mathbbm{P}_0(X \in A^\pi_k)\Psi^\pi_k\bigg)
\end{equation*}
Since $\mathbbm{E}_{P_0}[\dot{\Psi}]= \sum^{K^\pi}_{k=1} \mathbbm{P}_0(X \in A^\pi_k)\big(\mathbbm{E}_{\mathbbm{P}_0}[D_{\Psi, \mathbbm{P}_0}(O)|X \in A^\pi_k] - \Psi^\pi_k\big)=0$, we must have that $\mathbbm{E}_{\mathbbm{P}_0}[D_{\Psi, \mathbbm{P}_0}(O)|X \in A^\pi_k] - \Psi^\pi_k=0$ for all $k$, so (not surprisingly)
\begin{equation}
\hat{\Psi}^\pi_k=\frac{1}{|\{i: X_i \in A^\pi_k\}|}\sum_{i: X_i \in A^\pi_k} D_{\Psi, \mathbbm{P}_0}(O_i)
\end{equation}
is the oracle unbiased estimator for $\Psi^\pi_k$. 
It must also be the unique efficient estimator, since if it was not, and there existed a regular estimator with lower asymptotic variance, say $\tilde{\Psi}^\pi_k$ for each $k$, the estimator $\tilde{\Psi}= \sum^{K^\pi}_{k=1}\frac{|\{i: X_i \in A^\pi_k\}|}{n} \tilde{\Psi}^\pi_k$ would be the efficient estimator for the population mean, which is a contradiction since $\hat{\Psi}$ is the unique asymptotically efficient estimator.  
\end{proof}

\begin{remark}[Relationship to IF-calculus]
Albeit more involved, similar results can be obtained also when deriving the EIF-based estimator using rules of simple calculus. Consider the simple case of the conditional mean, i.e. $\psi(X)=\mathbbm{E}_{\mathbbm{P}_0}[Y|X=x]$. Since $\Psi_{A^\pi_k} = \mathbbm{E}_{
\mathbbm{P}_0}[Y|X \in A^\pi_k] = \frac{\mathbbm{E}_{\mathbbm{P}_0}[\mathbbm{1}\{X \in A^\pi_k\}Y]}{\mathbbm{E}_{\mathbbm{P}_0}[\mathbbm{1}\{X \in A^\pi_k\}]}$, the efficient influence function of the parameter  $\Psi_{A^\pi_k}$ can be shown to be $\dot{\Psi}_{A^\pi_k}= \frac{\mathbbm{1}\{X\in A^\pi_k\}}{\mathbbm{P}_0(X \in A^\pi_k)} \left( D_{\Psi, \mathbbm{P}_0}(O) - {\Psi}_{A^\pi_k}\right)$. When solving the estimating equation induced by setting $\sum^n_{i=1}\dot{\Psi}_{A^\pi_k}(O_i)=0$, we find that $\hat{\Psi}^\pi_k = \frac{\sum^n_{i=1}\frac{\mathbbm{1}\{X_i \in A^\pi_k\}}{\mathbbm{P}_0(X \in A^\pi_k)} D_{\Psi, \hat{\mathbbm{P}}}(O_i)}{\sum^n_{i=1}\frac{\mathbbm{1}\{X_i \in A^\pi_k\}}{\mathbbm{P}_0(X\in A^\pi_k)} }$. Estimating $\mathbbm{P}_0(X \in A^\pi_k)$ using the empirical distribution $\mathbbm{P}_n=\frac{1}{n} \sum^n_{i=1} \delta_{X_i}$ gives $\hat{\mathbbm{P}}(x \in A^\pi_k)= \frac{1}{n} \sum^n_{i=1} \mathbbm{1}\{x \in A^\pi_k\}$, and substitution leads to the same expression as in Proposition \ref{partprop}.
\end{remark}

\subsection{Regularity conditions for convergence in pseudo-outcome regression (Proposition 2)}\label{condthm1}

We assume the following pseudo-outcome regression set-up, as implied by Proposition 2 with standard regularity conditions as characterised by \cite{stone1980optimal}:  let $(D_\Psi, X)$ denote a pair of random random variables, with $D_\Psi\equiv D_{\Psi, \mathbbm{P}_0}(O)$ real-valued and $X \in \mathcal{X} \subseteq \mathbbm{R}^d$, where $\mathcal{X}$ denotes an open neighbourhood  around the origin. Let $\mathcal{G}_p$ denote the set of $p-1$ times continuously differentiable real valued functions $g$ on $\mathbbm{R}^d$ with $p$ bounded derivatives.  Assume $\psi_f$ is a fixed function, and consider the set of parameters $\mathcal{P} = \{\psi_f + g: g \in \mathcal{G} \}$, of which the true pseudo-outcome regression function $\psi(x)=\mathbbm{E}_{\mathbbm{P}_0}[D_\Psi|X=x]$ is an unknown member. Assume that the conditional variance of $D_\Psi$ is bounded on $\mathcal{X}$,  and that the density of $X$ is absolutely continuous and bounded away from zero and one on $\mathcal{X}$. The conditional distribution is assumed to be of the form $f(d_\Psi|x, \psi(x))\phi(dd_\Psi)$ with $\phi$ a measure on $\mathbbm{R}$. It is assumed that $f(d_\Psi|x, t)$ is strictly positive and jointly measurable in its arguments as they vary over their support. Further, $\int d_\Psi f(d|x, t) \phi(dd_\Psi) = \mathbbm{E}[d_\Psi|x, t]=t$ for $t$ in the open interval containing $\{\psi(x): \psi \in \mathcal{M} \text{ and } x\in \mathcal{X}\}$, and $f$ is twice continuously differentiable on the domain and fulfils some additional technical conditions on the log-likelihood. Then, the optimal minimax rate attainable for estimating $\psi(x)$ by generic nonparametric regression is $n^{-p/(2p+d)}$

\subsection{Result for error bounds for pseudo-outcome regression due to \cite{kennedy2020optimal}}\label{kennedythm}

Due to the importance of \cite{kennedy2020optimal}'s result for our rate-discussions, we restate the necessary assumptions as well as the theorem in slightly adapted format below:

\begin{assumption}\label{kennedyass} \textbf{Regularity of regression estimators}\\
We need mild two mild assumptions on the regularity of our second-stage regression estimators $\hat{\mathbbm{E}}_n$. $\hat{\mathbbm{E}}_n$ needs to satisfy that:
\begin{enumerate}
\item $\hat{\mathbbm{E}}_n(Y|X=x) + c = \hat{\mathbbm{E}}_n(Y+c|X=x)$ for any constant c
\item If $\mathbbm{E}[Y|X=x]= \mathbbm{E}[W|X=x]$ then 
\begin{equation*}
	\mathbbm{E}\left[\{\hat{\mathbbm{E}}_n[W|X=x] - \mathbbm{E}[W|X=x]\}^2\right] \asymp\mathbbm{E}\left[\{\hat{\mathbbm{E}}_n[Y|X=x] - \mathbbm{E}[Y|X=x]\}^2\right]
\end{equation*}
\end{enumerate}
\end{assumption}

\begin{otherthm} \textbf{Error bound for pseudo-outcome regression (Theorem 1, \cite{kennedy2020optimal})}\\
Let $O_0^n = ({O_{01}, \ldots, O_{0n}})$ and $O_1^n = ({O_{11}, \ldots, O_{1n}})$ denote two independent training and test samples, respectively,  which are all sampled from the same model $\mathbbm{P}_0$. Let $\hat{f}(o) = \hat{f}(o; O_0^n)$ be an estimate of a function $f(o)$ using only the training data $O_0^n$ and define $m(x)\equiv \mathbbm{E}_{\mathbbm{P}_0}[f(O)|X=x]$
Let $\hat{\mathbbm{E}}_n(Y|X=x)$ denote a generic estimator of the regression function $\mathbbm{E}[Y|X=x]$ using the test data $O_1^n$, where $(X_{1i}, Y_{1i}) \subseteq O_{1i}\text{, } i=1, \ldots, n$. Assume the estimator $\hat{\mathbbm{E}}_n$ satisfies
\begin{enumerate}
\item $\hat{\mathbbm{E}}_n(Y|X=x) + c = \hat{\mathbbm{E}}_n(Y+c|X=x)$ for any constant c
\item If $\mathbbm{E}[Y|X=x]= \mathbbm{E}[W|X=x]$ then 
\begin{equation*}
	\mathbbm{E}\left[\{\hat{\mathbbm{E}}_n[W|X=x] - \mathbbm{E}[W|X=x]\}^2\right] \asymp\mathbbm{E}\left[\{\hat{\mathbbm{E}}_n[Y|X=x] - \mathbbm{E}[Y|X=x]\}^2\right]
\end{equation*}
\end{enumerate}
 Let $\hat{m}(x) = \hat{\mathbbm{E}}_n[\hat{f}(O)|X=x]$ denote the regression of $\hat{f}(O)$ on the test samples, and let $\tilde{m}(x) =  \mathbbm{E}_n [f(O)|X=x]$ denote the corresponding oracle regression of $f(O)$ on $X$. Then we have that:
\begin{equation}
\mathbbm{E}\left[\{\hat{m}(x) - {m}(x)\}^2\right] \lesssim \mathbbm{E}\left[\{\tilde{m}(x) - {m}(x)\}^2\right] + \mathbbm{E}\{\hat{r}(x)^2\}
\end{equation}
where $\hat{r}(x)=\hat{r}(x;Z_0^n)\equiv \mathbbm{E}[\hat{f}(O)|X=x, O_0^n]-m(x)$.

For i.i.d. data, two independent samples can be arrived at by randomly splitting the sample in half. Further, the same bound holds when using cross-fitting (see e.g. \cite{chernozhukov2018double}) instead. 
\end{otherthm}

\subsection{Implications of doubly robust remainders for CATE estimation}\label{CATEImpl}
For the case of CATE estimation, \cite{kennedy2020optimal} (Theorem 2) showed that
\begin{equation*}
\{\mathbbm{E}[\bar{\psi}(x)|X=x, O_0^n] - \psi(x)\}^2 \leq \left(\frac{2}{\epsilon}\right)\{\pi(x) - \hat{\pi}(x)\}^2\left[\{\mu_1(x)-\hat{\mu}_1(x)\}^2 + \{\mu_0(x)-\hat{\mu}_0(x)\}^2\right]
\end{equation*}
where $\epsilon < \hat{\pi}(x) < 1-\epsilon$. Then, (modulo constants) the expectation of the RHS in (\ref{errorpsi}) can be written as $\mathbbm{E}\left[\{\hat{\pi}(x) - \pi(x)\}^2\right]\sum^1_{w=0}\mathbbm{E}\left[\{\mu_w(x)-\hat{\mu}_w(x)\}^2\right]$, if the propensity score and the outcome regressions are fit on separate samples \citep{kennedy2020optimal}. 

For a simple plug-in estimator, $\hat{\mu}_1 - \hat{\mu}_0$, on the other hand, the remainder is given by
\begin{equation*}
 (\hat{\mu}_1(x) - \hat{\mu}_0(x) - {\mu}_1(x) - {\mu_0}(x))^2 \leq 2 * \left( \{\mu_1(x)-\hat{\mu}_1(x)\}^2 + \{\mu_0(x)-\hat{\mu}_0(x)\}\right)^2
\end{equation*}
as $(a+b)^2 \leq 2 (a^2 + b^2)$.

Because the product term with the propensity score is missing, the rate of convergence of this plug-in estimator will be dominated by the slower-converging of the two regression terms, while the convergence of the IF-based estimator is dominated by the slower of the oracle rate and the remainder term. Since the the remainder term always decays faster than the slower of the two regression terms (since it is interacted with the propensity score term), the IF-based estimator is preferable over the plug-in estimator whenever CATE is a simpler function than the more complex of the two potential outcomes functions.

\section{Comparison of orthogonalised regression estimators and AIPW estimators for treatment effect estimation}\label{wadiscussion}
In this note we briefly compare the properties and assumptions underlying orthogonalised regression estimators of treatment effects (as used in, for example, \cite{Athey2019}'s causal trees and \cite{Chernozhukov2018}'s algorithms) and EIF-based AIPW estimators (which are the basis of this paper as well as, e.g. \cite{kennedy2020optimal} and \cite{lee2017doubly}). 
\subsection{Orthogonalised regression estimators}
Orthogonalised regression estimators appear mainly in the econometrics literature and naturally follow from \cite{robinson1988root}'s root-n consistent approach to semiparametric regression in the model
\begin{equation}\label{semireg}
    \mathbbm{E}_{\mathbbm{P}}[Y|W, X] = \beta'W + \eta(X) 
\end{equation}
where $(Y, W, X) \in \mathbbm{R} \times \mathbbm{R}^p \times \mathbbm{R}^q$, $X$ and $W$ are non-overlapping sets of covariates (and $W$ is not perfectly predictable from $X$) and $\beta$ is a regression parameter of interest while $\eta(X)$ has unspecified form. 
\cite{robinson1988root} and \cite{chernozhukov2018double} show that no-intercept ordinary least squares regression (OLS) in the residualised/orthogonalised model
\begin{equation}\label{orthoreg}
    Y- \mathbbm{E}_{\mathbbm{P}}[Y|X] = \beta'(W-  \mathbbm{E}_{\mathbbm{P}}[W|X]) + U
\end{equation}
with error-term $U$ induces a semiparametrically efficient estimator for $\beta$ if $\mathbbm{E}_\mathbbm{P}[U|W, X]=0$ (the errors are exogenous), the model \ref{semireg} is correctly specified and the error-terms $U$ are homoskedastic. 

The orthogonalised local regression estimator in \cite{Athey2019} for CATE $\tau(x)$ of treatment $W \in \{0,1\}$ is based on a local version of (\ref{orthoreg}), using a local estimation equation of the form 
\begin{equation}\label{eeathey}
    \arg \min_{\tau(x)} \sum^n_{i=1} \alpha_x(X_i)((Y_i - \hat{\mu}(X_i)) - (W_i - \hat{\pi}(X_i))\tau(x))^2
\end{equation}
with forest-based kernel-weights $\alpha_x(X_i)$ and forest-based out-of-bag estimates $\hat{\mu}(x) =\hat{\mathbbm{E}}[Y|X=x]$ and $\hat{\pi}(x)=\hat{\mathbbm{E}}[W|X=x]$. By analogy with (\ref{orthoreg}) this estimator is efficient if (i) the local treatment effect model is correctly specified, i.e. the treatment effect is additive as in semiparametric regression and $\tau(x)=\tau$ is (approximately) constant in a neighborhood around $x$ and (ii) error terms are locally homoskedastic. 
\subsection{AIPW estimators}
In this paper, we consider a very general model of treatment effects, which is less restrictive than that implied by semiparametric regression -- namely we let $\mathbbm{E}_\mathbbm{P}[Y|X, W]=f(X, W)$ so that treatment effects are not necessarily additive and error terms could be heteroskedastic. As discussed in section \ref{sectionrct}, the AIPW estimator for the ATE is given by
\begin{equation*}
D_{\Psi, \mathbbm{P}}(O)=  \left(\frac{W}{\pi(X)}- \frac{(1-W)}{1-\pi(X)}\right) Y + \left[\left(1 - \frac{W}{\pi(X)}\right) \mu_1(x)-\left(1 - \frac{1-W}{1-\pi(X)}\right)\mu_0(X)\right]
\end{equation*}
This estimator is the nonparametrically efficient estimator and thus efficient when the assumptions listed above do not hold \citep{chernozhukov2017double}. 

As we have shown in section \ref{maincontrib}, when used as our pseudo-outcome, this estimator approaches an approximately efficient nonparametric estimator for CATE. By using \cite{Athey2019}'s random forests for estimation of our pseudo-outcome regression model, we are solving the local estimation equation
\begin{equation}
    \arg \min_{\tau(x)} \sum^n_{i=1} \alpha_x(X_i)(D_{\Psi, \hat{\mathbbm{P}}}(O)-\tau(x))^2
\end{equation}
where the dependence on $\hat{\mathbbm{P}}$ signifies that we rely on forest-based out-of-bag estimates for the nuisance parameters $\pi(x)$ and $\mu_w(x)\text{, } w \in \{0,1\}$.

\subsection{Implications and expectations for empirical performance}
In terms of asymptotic (large sample) performance, the discussion in the previous two sections allows to make straightforward predictions for the relative performance of CATE-estimators based on orthogonalised regression and AIPW: if the assumptions on data-generating process (model specification and homoskedasticity) hold, then CATE-estimators based on orthogonalised regression are asymptotically efficient because some structure is correctly specified, reducing the difficulty of the estimation problem. When considering finite (small) sample performance, we expect that there is a tradeoff between two factors which may impact the relative empirical performance: On the one hand, the orthogonalised regression estimator needs one less plug-in estimator as it relies on only $\mu(x) = \mathbbm{E}_\mathbbm{P}[Y|X=x]$ instead of the two regression functions $\mu_w(x)$, which may result in better  finite sample performance. On the other hand, the smaller the sample, the less likely that the causal forest can identify neighborhoods in which the linearly additive and locally constant treatment effects specification encoded in (\ref{eeathey}) holds. 

\section{Experimental set-up}
\subsection{Simulation study 1}\label{expsetup}
\subsubsection{Data generating processes}
In our experiments, we revisit the example setting used in \cite{kennedy2020optimal}, which is based on a difficult piecewise polynomial baseline effect function $\mu_0(x)$ from \cite{gyorfi2006distribution}, while the treatment effect $\tau(x)=\tau=0$ is not only constant but also zero. This set-up illustrates very well how plug-in bias affects data-adaptive target function estimates even when the data is only one-dimensional.

Thus, we use the following piecewise polynomial baseline model specification:
\begin{equation}
\begin{split}
\mu_0(x) = 0.5 \times \mathbbm{1}\{x \leq -0.5\}(x+2)^2 + (x/2 -0.875) \times \mathbbm{1}\{-0.5<x\leq 0\} + \\\mathbbm{1}\{0<x\leq 0.5\}(-5(x-0.2)^2+1.075) + \mathbbm{1}\{x > 0.5\}(x + 0.125)
\end{split}
\end{equation}

For continuous outcomes, we generate observations $Y_i$ using the same model as \cite{kennedy2020optimal}: We simulate inputs
\begin{equation}\label{simu}
\begin{split}
X_i \sim Unif([-1, 1])\\
W_i \sim Ber(\pi(X_i))
\end{split}
\end{equation}
and outcomes as
\begin{equation}
Y_i = W_i * \tau(X_i) + \mu_0(X_i) + \epsilon(X_i)
\end{equation}
with 
\begin{equation*}
\epsilon(X_i) \sim \mathcal{N}(0, 0.2 -0.1\times cos(2\pi\times X_i))
\end{equation*}

Further, we use different degrees of selection bias as represented by the propensity score $\pi(X_i)$. For the first setting in Figures \ref{experiment1} and \ref{experiment2}, we use $\pi(x) = \pi = 0.5$, i.e. full randomization. For the second setting in Figures \ref{experiment1} and \ref{experiment2}, as well as the results presented in Table 2, we use
\begin{equation}
\pi(x) = 0.1 + 0.8*\mathbbm{1}\{x>0\}
\end{equation}
the propensity score used by \cite{kennedy2020optimal}. For the final setting in Figure \ref{experiment1}, we use
\begin{equation}
\pi(x) = 0.5 + 0.5 \times b \times \frac{|x|}{2}
\end{equation}
i.e. there is selection bias that selects individuals with higher absolute values in covariates into treatment, which gets stronger as $b$ increases in $[0, 1)$, but feed the model a propensity score of $0.5$ as an input.

In Figure \ref{examplesim} we illustrate the most extreme selection bias setting used in this paper (setting 2), which corresponds to the setting used in \cite{kennedy2020optimal}.

\begin{figure}[!htb]
\centering
\includegraphics[width=0.75\textwidth]{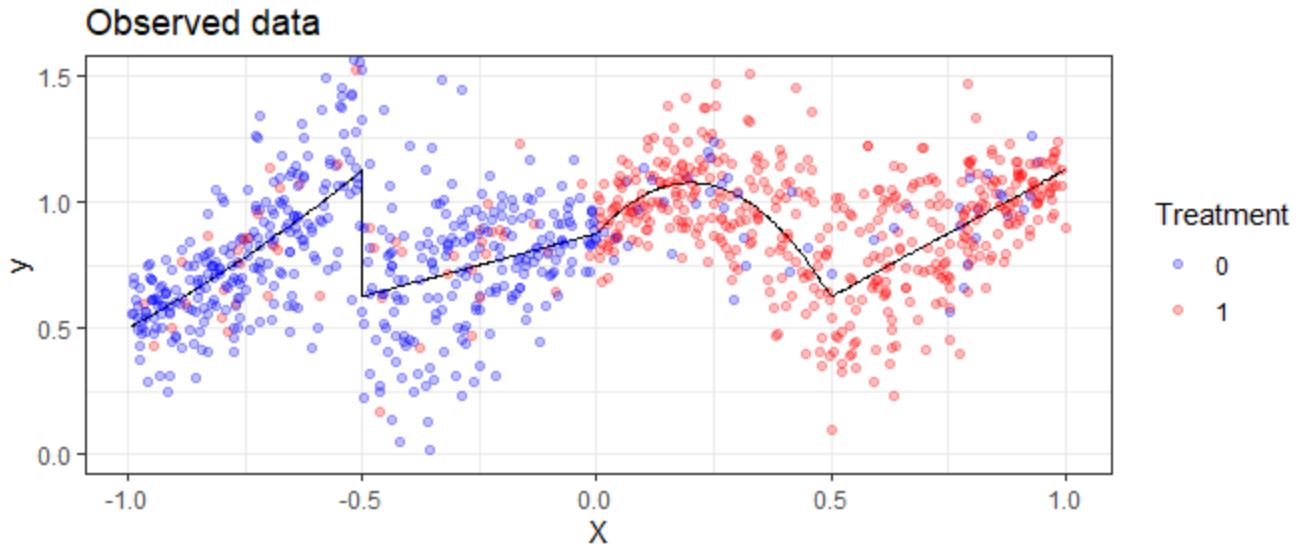}
\caption{Simulated observed data (dots) and underlying baseline function (black line) for the second setting, for $n=1000$}\label{examplesim}
\end{figure}

For the binary outcome data, we keep the baseline simulations (\ref{simu}) but for outcome simulations instead use 
\begin{equation}
Y_i \sim {Ber}\left(\frac{W_i * \tau(X_i) + \mu_0(X_i)}{1.5}\right) 
\end{equation}
where scaling by $1.5$ is necessary to ensure probabilities between $0$ and $1$.

\subsubsection{Estimation methods}
For all settings, we use the same generic default estimator: whenever the outcome variable is continuous, we use an adaptive smoothing spline, as implemented in base R by the function `smooth.spline'. Whenever the outcome is binary, we use an adaptive LogisticGAM based on a natural cubic spline basis, using the default implementation in the python package PyGAM \citep{pygam}. Here, adaptive indicates that the spline smoothness penalty $\lambda$ is chosen adaptively while fitting, by the default implementation in `smooth.spline' and by grid search in LogisticGAM. 

We present results that are averaged over 500 independent simulations using the same DGP, where the MSE is computed on a test-set of 1000 independent hold-out observations in each simulation. When training on very small samples, boundary bias sometimes resulted in highly unusual values for the MSE. Therefore, runs resulting in a MSE $>1000$ were discarded for all methods, to avoid contamination of the mean MSE, which was typically below 1. 

\subsection{Simulation study 2}\label{expsetup2}
For the second simulation study based on the toy-examples in \cite{Athey2019}, we simulate inputs as:
\begin{equation}\label{simu2}
\begin{split}
X_i \sim Unif([0, 1]^{10})\\
W_i \sim Ber(\pi(X_i))
\end{split}
\end{equation}
and outcomes as
\begin{equation}
Y_i = W_i * \tau(X_i) + \mu_0(X_i) + \epsilon_i \text{ with }
\epsilon_i \sim \mathcal{N}(0, 1)
\end{equation}
with $\pi(x)$, $\mu_0(x)$ and $\tau(x)$ varying as given in Table \ref{atheyres}. Regression forests and causal forests are implemented using the R package \url{grf} \citep{Athey2019}.

\end{document}